\documentclass[english]{article}
\usepackage{amsmath,amsthm} 
\usepackage{graphicx}
\usepackage{amssymb}
\usepackage{enumitem}
\usepackage{natbib} 
\setcitestyle{round}
\usepackage{bussproofs}
\usepackage[letterpaper]{geometry} 
\usepackage{appendix}
\usepackage{cancel}
\usepackage{booktabs}
\usepackage{float}

\DeclareSymbolFont{symbolsC}{U}{txsyc}{m}{n}
\DeclareMathSymbol{\boxright}{\mathrel}{symbolsC}{128}

\usepackage{thmtools,thm-restate}

\usepackage{tikz}
\usepackage{pgf} 
\usetikzlibrary{patterns,automata,arrows,shapes,snakes,topaths,trees,backgrounds,positioning,through,calc}
\usepackage{setspace}
\usepackage{multicol}
\usepackage{multirow}
\usepackage{graphicx}
\usepackage{xcolor}

\theoremstyle{definition}
\newtheorem{theorem}{Theorem}[section]

\newtheorem{proposition}[theorem]{Proposition}
\newtheorem{corollary}[theorem]{Corollary}
\newtheorem{definition}[theorem]{Definition}

\newtheorem{lemma}[theorem]{Lemma}
\newtheorem{fact}[theorem]{Fact}
\newtheorem{example}[theorem]{Example}

\newtheorem{remark}[theorem]{Remark}

\usepackage{xcolor}
\definecolor{medgreen}{rgb}{0.0, 0.75, 0.0}

\makeatletter  

\usepackage{stmaryrd}  
\usepackage{comment}
\usepackage{hyperref}
\hypersetup{colorlinks=true, citecolor=blue, linkcolor=blue}

\newcommand{\margin}{\mathcal{M}}

\makeatother

\usepackage{babel} 
\begin{document} 

\title{Characterizations of voting rules based on majority margins}

\date{}

\author{Yifeng Ding$^*$, Wesley H. Holliday$^\dagger$, and Eric Pacuit$^\ddagger$\\ 
{\small $*$ Peking University {\normalsize (\href{mailto:yf.ding@pku.edu.cn}{yf.ding@pku.edu.cn})}} \\
 {\small $\dagger$ University of California, Berkeley {\normalsize (\href{mailto:wesholliday@berkeley.edu}{wesholliday@berkeley.edu})}} \\
 {\small $\ddagger$ University of Maryland {\normalsize (\href{mailto:epacuit@umd.edu}{epacuit@umd.edu})}}}
\maketitle

\begin{abstract} In the context of voting with ranked ballots, an important class of voting rules is the class of \textit{margin-based}  rules (also called \textit{pairwise} rules). A voting rule is margin-based if whenever two elections generate the same head-to-head margins of victory or loss between candidates, the voting rule yields the same outcome in both elections. Although this is a mathematically natural invariance property to consider, whether it should be regarded as a normative axiom on voting rules is less clear. In this paper, we address this question for voting rules with any kind of output, whether a set of candidates, a ranking, a probability distribution, etc. We prove that a voting rule is margin-based if and only if it satisfies some axioms with clearer normative content. A key axiom is what we call Preferential Equality, stating that if two voters both rank a candidate $x$ immediately above a candidate $y$, then either voter switching to rank $y$ immediately above $x$ will have the same effect on the election outcome as if the other voter made the switch, so each voter's preference for $y$ over $x$ is treated equally. 
\end{abstract}

\tableofcontents

\section{Introduction}\label{Intro}

In elections in which voters submit rankings of the candidates, one of the most natural questions to ask is how one candidate $x$ performed head-to-head against another candidate $y$. Did more voters rank $x$ above~$y$ than ranked $y$ above $x$? Assuming the answer is `yes', how many more? That number is $x$'s \textit{margin of victory} over $y$ and $y$'s \textit{margin of loss} to $x$. Not only is it natural to inquire into the margins of victory or loss between candidates, but also this is the \textit{only} information about the election that is needed to determine the election outcome according to many well-known voting rules. Examples include many Condorcet voting rules,\footnote{\label{CondorcetRules}Margin-based Condorcet voting rules can be found in, e.g., \citealt{Simpson1969}, \citealt{Kramer1977}, \citealt{Tideman1987}, \citealt{Dutta1999}, \citealt{Schulze2011}, \citealt{Fernandez2018}, \citealt{HP2020}, and \citealt{HP2023b}. All so-called C1 voting rules (\citealt{Fishburn1977}) are also margin-based. For the distinction between margin-based rules and C2 rules, see below in the main text.} as well as the Borda Count, at least in its standard formulation for linear orders (see \citealt[p.~27]{Zwicker2016}). Non-examples include Plurality and Instant Runoff Voting, whose winners cannot be determined just from the margin information.

To state the point more formally, fix infinite sets $\mathcal{V}$ and $\mathcal{X}$ of \textit{voters} and \textit{candidates}, respectively. A \textit{profile} is a function $\mathbf{P}:V\to \mathcal{O}(X)$ where $V$ is a nonempty finite subset of $\mathcal{V}$, $X$ is a nonempty finite subset of $\mathcal{X}$ with $|X|\geq 2$, and $\mathcal{O}(X)$ is the set of all strict weak orders on $X$;\footnote{Recall that a \textit{strict weak order} on $X$ is a binary relation $P$ on $X$ that is asymmetric and negatively transitive (i.e., if $(a,b)\not\in P$ and $(b,c)\not\in P$, then $(a,c)\not\in P$), which implies that $P$ is transitive.} let us write $V(\mathbf{P})$ for $V$ and $X(\mathbf{P})$ for $X$.\footnote{Allowing the set of voters to vary across profiles will be essential for our arguments. By contrast, all of our arguments would go through in a setting in which every profile has the same set of candidates.} Given a profile $\mathbf{P}$ and candidates $x,y\in X(\mathbf{P})$, we define functions $\#_\mathbf{P}$ and $\mathcal{M}_\mathbf{P}$ from $X(\mathbf{P})^2$ to $\mathbb{Z}$ as follows:
\begin{eqnarray}
\#_\mathbf{P}(x,y)&=&|\{i\in V(\mathbf{P})\mid (x,y)\in \mathbf{P}(i)\}| \label{HashDef}\\
\mathcal{M}_\mathbf{P}(x,y) &= & \#_\mathbf{P}(x,y) - \#_\mathbf{P}(y,x). \label{MDef}
\end{eqnarray}

\noindent A \textit{voting rule} is a function $F$ from some set $\mathrm{dom}(F)$ of profiles to some nonempty set. Whether the voting rule outputs a set of candidates, a binary relation on the set of candidates, a probability distribution, etc., will not matter.

\begin{definition}\label{MarginBased} A voting rule $F$ is \textit{margin-based} if for any $\mathbf{P},\mathbf{Q}\in \mathrm{dom}(F)$, if $\mathcal{M}_\mathbf{P}=\mathcal{M}_\mathbf{Q}$, then $F(\mathbf{P})=F(\mathbf{Q})$.
\end{definition}
\noindent Margin-based rules are called \textit{pairwise}  in \citealt{Brandt2018a} and \textit{C1.5} in \citealt{DeDonder2000}.

Definition \ref{MarginBased} is a mathematically natural invariance property to consider, but whether it should be regarded as a normative axiom on voting rules is less clear. In this paper, we address this question by proving that a voting rule is margin-based if and only if it satisfies some axioms with clearer normative content. A key axiom is what we call Preferential Equality (Definition \ref{PrefEq}), stating that if two voters both rank a candidate $x$ immediately above a candidate $y$, then either voter switching to rank $y$ immediately above $x$ will have the same effect on the election outcome as if the other voter made the switch, so each voter's preference for $y$ over $x$ is treated equally. This implies (Lemma~\ref{PrefEq3}) that if there are some voters who rank $x$ immediately above $y$, then for any partition of those voters into groups $I$ and $J$ of equal size, the voters in $I$ switching to rank $y$ immediately above $x$ will have the same effect on the election outcome as if the voters in $J$ made the switch, so both groups' preferences for $y$ over $x$ are treated equally.

Preferential Equality is clearly satisfied by margin-based voting rules. For an example violation by a voting rule that is not margin-based, consider Instant Runoff Voting (IRV) and the following profile $\mathbf{P}$ representing a possible election between a Democrat $D$, moderate Republican $M$, and more extreme Republican $R$:
\begin{center}
\begin{tabular}{cccc}
37\% & 3\% &  32\% & 28\% \\
\hline
$D$ & $D$ & $R$ & $M$\\
$M$ & $R$ & $M$ & $R$\\
$R$ & $M$ & $D$ & $D$
\end{tabular}
\end{center}
In this election, IRV elects $R$, because $M$ is eliminated in the first round due to having the fewest first-place votes, and then $R$ beats $D$ in the second round. 
Now consider the 3\% of Democratic voters who rank $R$ immediately above $M$, as well as 3\% out of the 32\% of Republican voters who also rank $R$ immediately above $M$. In particular, consider the effect of either one of these groups switching to rank $M$ immediately above $R$:
\begin{center}
\begin{tabular}{cccc}
\textcolor{medgreen}{40\%} & \textcolor{red}{0\%} &  32\% & 28\% \\
\hline
$D$ & $D$ & $R$ & $M$\\
\textcolor{medgreen}{$M$} & \textcolor{red}{$R$} & $M$ & $R$\\
\textcolor{medgreen}{$R$} & \textcolor{red}{$M$} & $D$ & $D$
\end{tabular}\qquad\qquad \begin{tabular}{cccc}
37\% & 3\% &  \textcolor{red}{29\%} & \textcolor{medgreen}{31\%} \\
\hline
$D$ & $D$ & \textcolor{red}{$R$} & \textcolor{medgreen}{$M$}\\
$M$ & $R$ & \textcolor{red}{$M$} & \textcolor{medgreen}{$R$}\\
$R$ & $M$ & $D$ & $D$
\end{tabular}
\end{center}
In the election on the left, IRV again elects $R$, because $M$ is eliminated in the first round due to having the fewest first-place votes, and then $R$ beats $D$ in the second round. However, in the election on the right, IRV elects $M$, because $R$ is eliminated in the first round due to having the fewest first-place votes, and then $M$ beats $D$ in the second round. This is a violation of Preferential Equality: starting from $\mathbf{P}$, it is possible for 3\% of Republican voters to change the election outcome by flipping $RM$ to $MR$, but it is not possible for 3\% of Democratic voters to change the election outcome by flipping $RM$ to $MR$. Moreover, we shall see (in Section~\ref{PrefEqSection}) that there are many real political elections in which similar violations of Preferential Equality are possible for IRV.

There may be contexts in which treating voters equally does not require that a voting rule satisfy the axiom we have called `Preferential Equality'. For example, if the goal of the voting process is to maximize social utility,\footnote{For arguments against taking preference intensities into account in collective choice, see \citealt[\S~2.2]{Schwartz1986}.} and the designer of the voting rule somehow knows that for any election to which the rule will be applied and any two voters, $i$ and $j$, the utility difference for $i$ between the candidate $i$ ranks 1st and the candidate $i$ ranks 2nd is greater than the utility difference for $j$ between the candidate $j$ ranks 2nd and the candidate $j$ ranks 3rd,\footnote{This of course assumes the possibility of interpersonal comparisons of utility,  contrary to Arrow's \citeyearpar{Arrow1963} view that ``interpersonal comparison of utilities has no meaning and, in fact, that there is no meaning relevant
to welfare comparisons in the measurability of individual utility'' (p.~9).} then one would not treat as equivalent $i$'s switching $x$ and $y$ in the 1st and 2nd positions of $i$'s ranking and $j$'s switching $x$ and $y$ in the 2nd and 3rd positions of $j$'s ranking. As another example, if one knows that voters are confident about their 1st choice but then flip a coin to determine who to rank 2nd vs.~3rd, one need not treat $i$'s switch and $j$'s switch as equivalent. However, in contexts without information about systematic differences in utility gaps or levels of noise between different positions in rankings, considerations of fairness and equal respect for voters motivate Preferential Equality.

Our second key axiom is what Saari \citeyearpar{Saari2003} calls Neutral Reversal. This axiom states that adding to any profile a pair of reversed linear orders does not change the election outcome. The idea is that two voters with fully reversed preferences balance each other out, so the election outcome is unchanged. This can be motivated by viewing an election as follows: by casting a ballot, a voter makes a contribution---independent of the contributions of other voters---to how well or poorly each candidate performs in the election; reversed rankings make opposite contributions; and the contribution of two voters is the sum of the contributions of each voter individually. Admittedly, one could instead adopt a holistic or super-additive view according to which, e.g., a pair of reversed rankings favors the middle-ranked candidates in the pair over the more ``polarizing'' candidates at the top and bottom of the respective rankings.\footnote{Thanks to an anonymous referee for raising this point.} For our purposes, the point is not that there is a decisive argument for Neutral Reversal but rather that there is a view motivating Neutral Reversal that does not presuppose that a voting method should be margin-based. The motivation for Neutral Reversal above is consistent with adopting voting methods such as Positive/Negative voting (\citealt{Lapresta2010}, \citealt{Heckelman2020}), where each voter awards $1$ point to their top-ranked candidate and $-1$ to their bottom-ranked candidate, or a hybrid between IRV and the Coombs method (\citealt{Coombs1964}, \citealt{Grofman2004}) that iteratively eliminates the candidate(s) with the lowest ``score'', defined as their number of first-place votes minus their number of last-place votes. Neither of those methods is margin-based. On the other hand, it is easy to see that all margin-based voting rules satisfy Neutral Reversal, while Plurality and IRV do not.

Our first main theorem (Theorem \ref{LinChar}) is that for the domain of linear profiles, in which each voter submits a complete linear order of the candidates, a voting rule is margin-based if and only if it satisfies Preferential Equality and Neutral Reversal. If we add a weak assumption of Homogeneity (\citealt{Smith1973}) of the voting rule, we can replace Neutral Reversal with the less controversial axiom of Block Invariance (\citealt{HP2025}), according to which the output of the voting rule remains unchanged if we add exactly one copy of each linear order of the candidates (Theorem \ref{BlockInvarianceThm}). Homogeneity also allows us to characterize margin-based rules on the domain of all profiles, provided we strengthen Preferential Equality to account for breaking ties in rankings (Theorem~\ref{StrictWeakChar}).

There is a subtle but important distinction between margin-based voting rules and a larger class of what we will call \textit{head-to-head} voting rules, which are in turn closely related to what Fishburn~\citeyearpar{Fishburn1977} called the class of C2 voting rules. Given a profile $\mathbf{P}$, let 
\begin{equation}
    \mathcal{H}(\mathbf{P}) =(\#_\mathbf{P}, |V(\mathbf{P})|),\label{Hdef}
\end{equation}
where $\#_\mathbf{P}$ is the function defined in (\ref{HashDef}) and $|V(\mathbf{P})|$ is the number of voters in $\mathbf{P}$. Thus, $\mathcal{H}(\mathbf{P})$ tells us not only how many voters  prefer any $a$ to any $b$ but also the total number of voters, from which we can determine how many voters are indifferent between $a$ and $b$. 

\begin{definition}\label{H2H} A voting rule $F$ is a  \textit{head-to-head} voting rule if for any  $\mathbf{P},\mathbf{P}'\in\mathrm{dom}(F)$, if $\mathcal{H}(\mathbf{P})=\mathcal{H}(\mathbf{P}')$, then $F(\mathbf{P})=F(\mathbf{P}')$. The rule $F$ is \textit{C2} if for any  $\mathbf{P},\mathbf{P}'\in\mathrm{dom}(F)$, if $\#_\mathbf{P}=\#_{\mathbf{P}'}$, then $F(\mathbf{P})=F(\mathbf{P}')$.\footnote{Although this was Fishburn's \citeyearpar{Fishburn1977} original definition of C2 (except that, as in much of the literature, we do not exclude C1 methods from the class of C2 methods as Fishburn did), one could argue that since Fishburn's paper was in the context of profiles of linear orders, his paper did not determine a unique extension of the C2 concept to profiles allowing ties.} 
\end{definition}
\noindent C2 voting rules are trivially head-to-head, and since the margins are computed from $\#$, all margin-based rules are C2. An example of a C2 voting rule that is not margin-based is the \textit{weak Pareto rule} that given a profile $\mathbf{P}$ outputs all the candidates who are not weakly Pareto-dominated by any other candidate in $\mathbf{P}$, where candidate $x$ weakly Pareto-dominates candidate $y$ if 
some voter prefers $x$ to $y$ and no voter prefers $y$ to $x$, i.e., $\#_\mathbf{P}(x,y)>0$ and $\#_\mathbf{P}(y,x)=0$. Whether one candidate weakly Pareto-dominates another is not something we can tell just from $\mathcal{M}_\mathbf{P}$, since, e.g., $\mathcal{M}_\mathbf{P}(x,y)=2$ is consistent with not only $\#_\mathbf{P}(x,y)=2$ and $\#_\mathbf{P}(y,x)=0$ but also $\#_\mathbf{P}(x,y)=4$ and $\#_\mathbf{P}(y,x)=2$. An example of a head-to-head voting rule that is neither margin-based nor even C2 is the \textit{strict Pareto rule} that given a profile $\mathbf{P}$ outputs all the candidates who are not strictly Pareto-dominated by any other candidate in $\mathbf{P}$, where $x$ strictly Pareto-dominates $y$ if every voter prefers  $x$ to  $y$, i.e.,  $\#_\mathbf{P}(x,y)=|V(\mathbf{P})|$. Whether one candidate strictly Pareto-dominates another is not something we can tell just from $\mathcal{M}_\mathbf{P}$ or even $\#_\mathbf{P}$, since we need to know the total number of voters to determine if \textit{every} voter strictly prefers $x$ to $y$.\footnote{The strict Pareto rule falls into a class of rules that is intermediate between margin-based rules and the larger class of head-to-head rules, but incomparable with the class of C2 rules and the class of C2w rules, namely the class of rules that can be computed using just the margins and the number of voters, since $x$ strictly Pareto-dominates $y$ if and only if $\mathcal{M}_\mathbf{P}(x,y)=|V(\mathbf{P})|$. We do not know of any literature discussing this class of rules, which could perhaps be called the \textit{quasi-margin-based} rules.  The weak Pareto rule and the alternative version of the Minimax rule defined using $\#_\mathbf{P}$ in the main text below are not quasi-margin-based, though they are C2.} Of course, over the domain of linear profiles, the weak Pareto rule is equivalent to the strict Pareto rule and hence becomes C2. 

\begin{remark}\label{C2w} While we defined C2 in terms of $\#_\mathbf{P}$, which counts the number of voters who \textit{strictly} prefer $x$ to $y$, one might consider instead counting the number of voters who \textit{weakly} prefer $x$ to $y$. Where $\#^w_\mathbf{P}(x,y)=|\{i\in V(\mathbf{P})\mid (y,x)\not\in \mathbf{P}(i) \}|$, say that a voting rule $F$ is \textit{C2w} if for any  $\mathbf{P},\mathbf{P}'\in\mathrm{dom}(F)$, if $\#^w_\mathbf{P}=\#^w_{\mathbf{P}'}$, then $F(\mathbf{P})=F(\mathbf{P}')$. The notions of C2 and C2w are equivalent for voting rules defined on the domain of linear profiles, but they are not equivalent in general. All C2w voting rules are head-to-head since the $\#^w_\mathbf{P}$ function can be computed from the $\#_\mathbf{P}$ function together with $|V(\mathbf{P})|$ by the formula $\#^w_\mathbf{P}(x,y)=|V(\mathbf{P})| - \#_\mathbf{P}(y,x)$. In addition, all margin-based rules are C2w since $\mathcal{M}_\mathbf{P}(x,y) = \#^w_\mathbf{P}(x,y) - \#^w_\mathbf{P}(y,x)$. The above strict Pareto rule is C2w since $x$ strictly Pareto-dominates $y$ iff $\#^w_\mathbf{P}(y, x) = 0$. Analogously, the weak Pareto rule is C2 but not C2w, as shown by the following two (anonymized) profiles having the same $\#^w$ counts:
\begin{center}
\begin{tabular}{ccc}
$1$ & $1$  \\
 \hline
 $a$ & $a, b$ \\
 $b$ &   \\
\end{tabular}
\hspace{4em}
\begin{tabular}{ccc}
$2$ & $1$  \\
 \hline
 $a$ & $b$ \\
 $b$ & $a$ \\
\end{tabular}
\end{center}
Thus, the inclusion relations between head-to-head, C2, C2w, and margin-based rules can be summarized by the simple ``diamond'' diagram:

\begin{center}
\begin{tikzpicture} 
  \node[align=center] (hh) at (0,1.1) {head-to-head};
  \node[align=center] (c2) at (-2.2,0) {C2};
  \node[align=center] (c2w) at (2.2,0) {C2w};
  \node[align=center] (mb) at (0,-1.1) {margin-based};

  \draw[->, thick] (mb) -- (c2);
  \draw[->, thick] (mb) -- (c2w);
  \draw[->, thick] (c2) -- (hh);
  \draw[->, thick] (c2w) -- (hh);
\end{tikzpicture}
\end{center}
\end{remark}

We will also characterize when a head-to-head voting rule is margin-based, by way of characterizing C2 voting rules among head-to-head voting rules and characterizing margin-based voting rules among C2 voting rules. The normative importance of such characterizations largely has to do with questions about how to formulate common Condorcet voting rules when incomplete rankings (or ties in rankings) are allowed. Instead of measuring the strength of a majority preference for $x$ over $y$ in terms of the margin of victory of $x$ over $y$, some authors have considered measuring the strength of majority preference for $x$ over $y$ in other ways computable from $\#_\mathbf{P}$ (see \citealt[\S~2.1]{Schulze2011}). For example, although the Minimax voting rule (\citealt{Simpson1969}, \citealt{Kramer1977}) is now usually defined so as to output for a profile $\mathbf{P}$ the candidates $x$ who minimize the quantity $\mathrm{max}\{\mathcal{M}_\mathbf{P}(y,x)\mid y\in X(\mathbf{P})\}$, some have considered an alternative version that outputs the candidates $x$ who minimize $\mathrm{max}\{\#_\mathbf{P}(y,x)\mid y\in X(\mathbf{P}),\,\#_\mathbf{P}(y,x)> \#_\mathbf{P}(x,y)\}$. The latter ``winning votes'' version is C2 but not margin-based (or C2w). To see how these definitions can output different sets of candidates, consider the profile at the top of Figure~\ref{MarginsVsWinningVotes}; below the profile we show both the \textit{margin graph} of the profile, where an edge from candidate $x$ to candidate $y$ of weight $k$ indicates that $\mathcal{M}_\mathbf{P}(x,y)=k$, and the \textit{winning votes graph} of the profile, where here an edge from candidate $x$ to $y$ of weight $k$ indicates that $\#_\mathbf{P}(x,y)=k > \#_\mathbf{P}(y,x)$. The standard version of Minimax selects $c$ as the winner, whereas the winning votes version selects $a$ as the winner. 

Moreover, these distinctions are not merely academic. Examples  from real political elections with ranked ballots are shown in Figure~\ref{Govan}, where the standard version of Minimax selects Dornan, while the winning votes version selects Flanagan, and Figure~\ref{Minneapolis}, where the standard version selects Arab, while the winning votes version selects Worlobach; in each case, we only show candidates in the \textit{Smith set} (\citealt{Smith1973}), which is the smallest set of candidates such that every candidate in the set has a positive margin against every candidate outside the set. Table~\ref{FrequencyDifferentMinimaxWinners} shows the frequencies with which the two versions of Minimax disagree in four election databases, conditional on there being no Condorcet winner (both versions of Minimax elect the Condorcet winner whenever one exists).

\begin{figure}[h]
\begin{center}
\fbox{
\begin{minipage}{5in}

\begin{center}

\begin{tabular}{ccc}
$3$  & $4$ & $2$  \\
 \hline
 $b$     & $c$ & $a$ \\
 $a$ $c$ & $a$ & $b$ \\
         & $b$ & $c$  

\end{tabular}
\end{center}

\begin{center}
    \begin{tikzpicture}

\node at (6,-1.5) (A) {\textcolor{black}{$a$}}; 
\node at (8,.5) (B) {\textcolor{black}{$b$}}; 
\node at (10,-1.5) (C) {\textcolor{medgreen}{$\boldsymbol{c}$}}; 

\path[->,draw,thick] (A) to node[fill=white] {\textcolor{black}{3}} (B);
\path[->,draw,thick] (B) to node[fill=white] {\textcolor{black}{1}} (C);
\path[->,draw,thick] (C) to node[fill=white] {2} (A);

\end{tikzpicture}\qquad\quad \begin{tikzpicture}

\node at (6,-1.5) (A) {\textcolor{medgreen}{$\boldsymbol{a}$}}; 
\node at (8,.5) (B) {\textcolor{black}{$b$}}; 
\node at (10,-1.5) (C) {\textcolor{black}{$c$}}; 

\path[->,draw,thick] (A) to node[fill=white] {\textcolor{black}{6}} (B);
\path[->,draw,thick] (B) to node[fill=white] {\textcolor{black}{5}} (C);
\path[->,draw,thick] (C) to node[fill=white] {4} (A);

\end{tikzpicture}
\end{center} 
\end{minipage}
}
\end{center}
\caption{Above: a profile with ties (three voters rank $a$ and $c$ as tied). Below, left: the margin graph of the profile. Below, right: the winning votes graph of the profile.}\label{MarginsVsWinningVotes}
\end{figure}
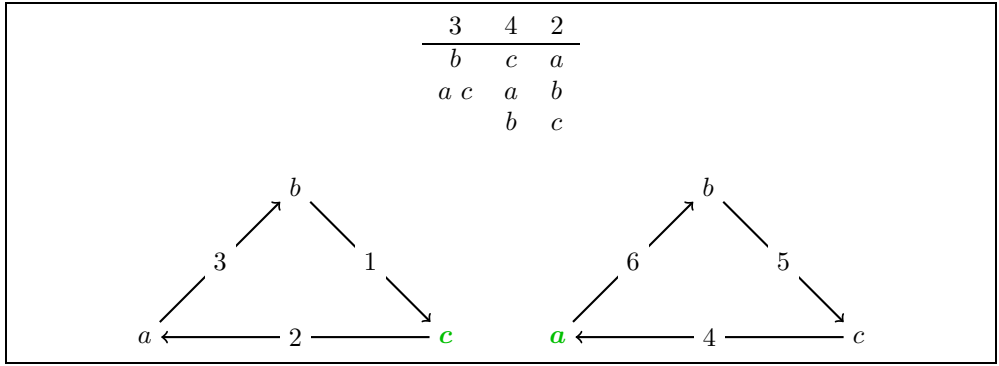

\begin{figure}[h]
\begin{center}
\fbox{
\begin{minipage}{5in}
\begin{center}

\begin{tabular}{ccccccccc}
$2521$ & $1645$ & $1051$ & $768$ & $681$ & $486$ & $381$ & $190$ & $180$\\
\hline
$h$ & $d$ & $f$ & $f$ & $d$ & $h$ & $h$ & $f$ & $d$\\
$d\ f$ & $f$ & $d$ & $d\ h$ & $f\ h$ & $d$ & $f$ & $h$ & $h$\\
  & $h$ & $h$ &   &   & $f$ & $d$ & $d$ & $f$\\
\end{tabular}
\vspace{.3in}

\begin{tikzpicture}

\node at (6,-1.5) (D) {\textcolor{medgreen}{\textbf{Dornan}}}; 
\node at (8,.5) (F) {\textcolor{black}{Flanagan}}; 
\node at (10,-1.5) (H) {\textcolor{black}{Hunter}}; 

\path[->,draw,thick] (D) to node[fill=white] {\textcolor{black}{$602$}} (F);
\path[->,draw,thick] (F) to node[fill=white] {\textcolor{black}{$86$}} (H);
\path[->,draw,thick] (H) to node[fill=white] {$21$} (D);

\end{tikzpicture}\qquad\quad \begin{tikzpicture}

\node at (6,-1.5) (D) {\textcolor{black}{Dornan}}; 
\node at (8,.5) (F) {\textcolor{medgreen}{\textbf{Flanagan}}}; 
\node at (10,-1.5) (H) {\textcolor{black}{Hunter}}; 

\path[->,draw,thick] (D) to node[fill=white] {\textcolor{black}{2992}} (F);
\path[->,draw,thick] (F) to node[fill=white] {\textcolor{black}{3654}} (H);
\path[->,draw,thick] (H) to node[fill=white] {$3578$} (D);

\end{tikzpicture}
\end{center}
\end{minipage}
}
\end{center}
\caption{Above: rankings in the 2007
Glasgow City Council election for Ward 5 (Govan), restricted to candidates in the Smith set. Below, left: the margin graph of the profile. Below, right: the winning votes graph of the profile.}\label{Govan}
\end{figure}
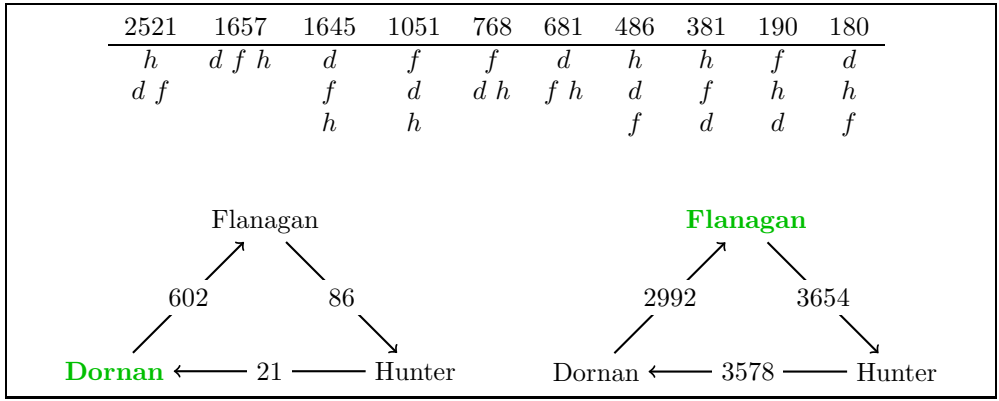

\begin{figure}[h]

\begin{center}
\fbox{
\begin{minipage}{5.75in}
\begin{center}

\begin{tabular}{ccccccccc}
$1572$ & $1299$ & $1177$ & $1086$ & $908$ & $822$ & $801$ & $758$ & $492$\\
\hline
$a$ & $w$ & $g$ & $w$ & $a$ & $g$ & $g$ & $a$ & $w$\\
$g\ w$ & $g$ & $w$ & $a$ & $g$ & $a\ w$ & $a$ & $w$ & $a\ g$\\
  & $a$ & $a$ & $g$ & $w$ &   & $w$ & $g$ &  \\
\end{tabular}

\vspace{.3in}

\begin{tikzpicture}
\node[minimum width=0.25in] at (6, -1.5) (a) {\textcolor{medgreen}{\textbf{Arab}}}; 
\node[minimum width=0.25in] at (10,-1.5) (w) {Worlobah}; 
\node[minimum width=0.25in] at (8, .5) (g) {Gordon};
\path[->,draw,thick] (a) to node[fill=white] {$225$} (g);
\path[->,draw,thick] (g) to node[fill=white] {$73$} (w);
\path[->,draw,thick] (w) to node[fill=white] {$15$} (a);
\end{tikzpicture}\qquad\quad \begin{tikzpicture}
\node[minimum width=0.25in] at (6, -1.5) (a) {Arab}; 
\node[minimum width=0.25in] at (10,-1.5) (w) {\textcolor{medgreen}{\textbf{Worlobah}}}; 
\node[minimum width=0.25in] at (8, .5) (g) {Gordon};
\path[->,draw,thick] (a) to node[fill=white] {$4324$} (g);
\path[->,draw,thick] (g) to node[fill=white] {$3708$} (w);
\path[->,draw,thick] (w) to node[fill=white] {$4054$} (a);

\end{tikzpicture}
\end{center}
\end{minipage}
}
\end{center}
\caption{Above: rankings in the 2021 Minneapolis City Council Ward 2 election, restricted to candidates in the Smith set. Below, left: the margin graph of the profile. Below, right: the winning votes graph of the profile.}\label{Minneapolis}
\end{figure}
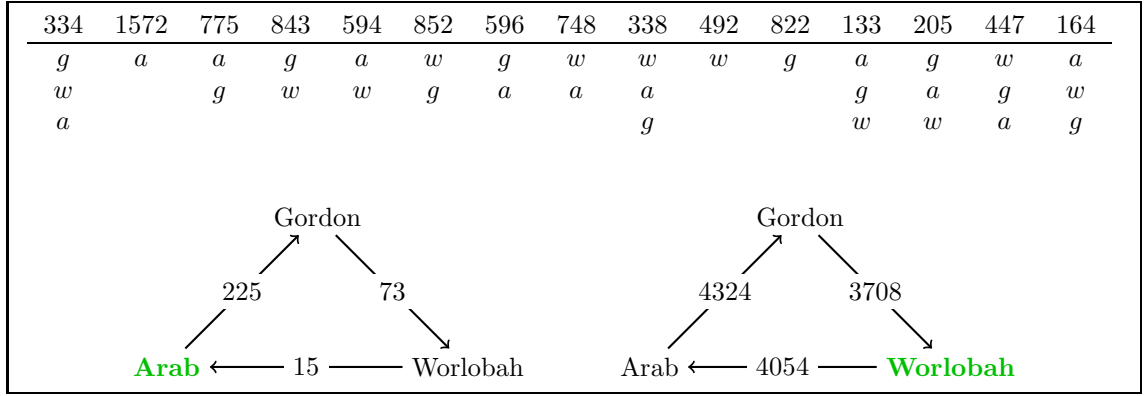

\begin{table}[H]
\begin{center}
\footnotesize
\begin{tabular}{lrrrrr}
\toprule
Dataset  & \# Profiles & \# Profiles & Avg. \#  & Avg. \# & Frequency of Different Winners \\
&&  with No CW &of Voters& of Cands. & Conditional on No CW\\
\midrule
Stable Voting Dataset & 974 &  205 & 8.80 & 4.77 &  0.03  \\
CIVS Dataset & 22,477 & 6,050 &  26.74 & 9.77 &  0.09  \\ 
PrefLib Politics (soi) & 364  & 1  & 44,577.96& 8.85 &  1.0  \\
\citealt{Otis2022} dataset & 458 & 2 & 74,840.35 & 6.10 &  0.5  \\
\bottomrule
\end{tabular}
\end{center}
\caption{Frequencies with which Minimax based on margins and Minimax based on ``winning votes'' select different winners, conditional on there not being a Condorcet winner (CW), in the elections from the Stable Voting Dataset 1-28-2026 (\citealt{HP2026data}), CIVS Dataset 2024-12-15 (\citealt{CIVS2024,Myers2024}), PrefLib Politics (soi files as of 12-20-2024) (\citealt{MatteiWalsh2013}), and \citealt{Otis2022}.  Note that we restricted to elections that are not ``test'' elections in the  CIVS Dataset.}\label{FrequencyDifferentMinimaxWinners}
\end{table}

All of this raises the question: what normative arguments are there for using the margin-based versions of voting rules? In the case of Minimax, previous work has identified reasons for favoring the margin-based version over the winning votes version: as Holliday and Pacuit \citeyearpar[Remark~3.11]{HP2025} observe, the former satisfies and the latter violates the axiom of Positive Involvement (stating that if a social choice correspondence selects a candidate $x$ in an initial profile $\mathbf{P}$, and $\mathbf{P}'$ is obtained from $\mathbf{P}$ by adding one voter who ranks $x$ uniquely first, then $x$ should still be selected in $\mathbf{P}'$). But what in general can be said about using margin-based rules as opposed to non-margin-based but still head-to-head rules? In what follows, we will provide answers to this question. 

The rest of the paper is organized as follows. In Section~\ref{LinProfs}, we characterize margin-based voting rules on the domain of linear profiles. In Section~\ref{StrictWeakProfs}, we characterize which homogeneous voting rules  on the domain of all profiles are margin-based. Finally, in Section~\ref{HeadToHead}, we characterize which head-to-head voting rules on the domain of all profiles are margin-based. We conclude in Section~\ref{Conclusion} with suggestions for future work. The Appendix gives analogues of theorems from Sections \ref{StrictWeakProfs}-\ref{HeadToHead} when we allow noncomparability, as distinguished from indifference, in voters' preference relations. Table~\ref{ThmSummary} summarizes our main results using terminology to be defined.

\begin{table}[h]
\begin{center}
\begin{tabular}{cccc}
\textbf{result} & \textbf{domain} & \textbf{extra assumption} & \textbf{axioms}\\
\hline
\multirow{2}{*}{Theorem \ref{LinChar}} & \multirow{2}{*}{linear profiles} & \multirow{2}{*}{none} & Preferential Equality \\
 & &  & Neutral Reversal \\
 \hline
 \multirow{2}{*}{Theorem \ref{BlockInvarianceThm}} & \multirow{2}{*}{linear profiles} & \multirow{2}{*}{Homogeneity} & Preferential Equality \\
 &&& Block Invariance \\
 \hline 
\multirow{3}{*}{Theorem \ref{StrictWeakChar}.\ref{StrictWeakChar2}} & \multirow{3}{*}{\textsf{LOBI}} & \multirow{3}{*}{Homogeneity} & Preferential Equality\\
&&& Tiebreaking Compensation \\
&&& Neutral Indifference \\
\hline
\multirow{2}{*}{Theorem \ref{StrictWeakChar}.\ref{StrictWeakChar1}} & \multirow{2}{*}{all profiles} & \multirow{2}{*}{Homogeneity} & Tiebreaking Compensation \\
&&& Neutral Indifference \\
\hline
\multirow{2}{*}{Theorem \ref{ModuloH2H}} & \multirow{2}{*}{all profiles} & \multirow{2}{*}{head-to-head} & Nonlinear Neutral Reversal \\
&&& Neutral Indifference
\end{tabular}
\end{center}
\caption{Summary of axiomatic characterizations of margin-based voting rules. Over the indicated domain, a voting method that satisfies the extra assumption is margin-based if and only if it satisfies all of the indicated axioms. Two additional theorems appear in the Appendix.}\label{ThmSummary}
\end{table}

A repository with code to generate some of the figures and tables in the paper, as well as to verify some facts, is available at \href{https://github.com/epacuit/majority_margins}{https://github.com/epacuit/majority\_margins}.

\section{Characterization for linear profiles}\label{LinProfs}

In this section, we begin with voting rules on the domain of \textit{linear} profiles in which each voter submits a complete linear order of the candidates. 

\subsection{Preferential Equality}\label{PrefEqSection}

Let us more formally define the axiom of Preferential Equality introduced in Section~\ref{Intro}. In the following, when we say that a voter $i$ ranks \textit{$x$ immediately above $y$}, we mean that $x\mathbf{P}_i y$ and for all ${z\in X(\mathbf{P})\setminus\{x,y\}}$, either $z\mathbf{P}_ix$ or $y\mathbf{P}_iz$.

\begin{definition}\label{PrefEq} A voting rule $F$  satisfies Preferential Equality if for any profile $\mathbf{P}\in\mathrm{dom}(F)$, candidates ${x,y\in X(\mathbf{P})}$, and two voters $i,j\in V(\mathbf{P})$ who both rank $x$ immediately above $y$, if $\mathbf{P}^i$ is obtained from $\mathbf{P}$ by $i$ switching to ranking $y$ immediately above $x$, and similarly $\mathbf{P}^j$ is obtained from $\mathbf{P}$ by $j$ switching to ranking $y$ immediately above $x$, then $F(\mathbf{P}^i)=F(\mathbf{P}^j)$ (or else $\mathbf{P}^i,\mathbf{P}^j$ do not both belong to $\mathrm{dom}(F)$).\end{definition}

Say that a domain of profiles is \textit{rich} if it is closed under the operation of flipping two adjacent candidates in a voter's ranking. Then we have the following convenient reformulation of Preferential Equality.

\begin{definition} A voting rule $F$ satisfies Preferential Compensation if for any ${\mathbf{P},\mathbf{P}'\in\mathrm{dom}(F)}$ and $x,y\in X(\mathbf{P})$, if $\mathbf{P}'$ is obtained from $\mathbf{P}$ by one voter who had $x$ immediately above $y$ in $\mathbf{P}$ switching to have $y$ immediately above $x$ in $\mathbf{P}'$, and another voter who had $y$ immediately above $x$ in $\mathbf{P}$ switching to have $x$ immediately above $y$ in $\mathbf{P}'$, then $F(\mathbf{P})=F(\mathbf{P}')$.
\end{definition}

\begin{lemma}\label{PrefEq2} If a voting rule $F$ satisfies Preferential Compensation, then $F$ satisfies Preferential Equality. Moreover, if $F$ satisfies Preferential Equality and its domain is rich, then $F$ satisfies Preferential Compensation.
\end{lemma}

\begin{proof} Assume $F$ satisfies Preferential Compensation. To show that $F$ satisfies Preferential Equality, let $\mathbf{P}$, $\mathbf{P}^i$, and $\mathbf{P}^j$ be as in Definition~\ref{PrefEq}, and assume $\mathbf{P}^i,\mathbf{P}^j\in\mathrm{dom}(F)$. Then $\mathbf{P}^i$ and $\mathbf{P}^j$ are related to each other exactly as $\mathbf{P}$ and $\mathbf{P}'$ are related in Preferential Compensation, so $F(\mathbf{P}^i)=F(\mathbf{P}^j)$. 

Now assume $F$ satisfies Preferential Equality and has a rich domain. To show that $F$ satisfies Preferential Compensation, suppose $\mathbf{P}, \mathbf{P}' \in \mathrm{dom}(F)$ are such that $\mathbf{P}'$ is obtained from $\mathbf{P}$ by voter $i$ who had $x$ immediately above $y$ in $\mathbf{P}$ switching to have $y$ immediately above $x$ in $\mathbf{P}'$, and voter $j$ who had $y$ immediately above $x$ in $\mathbf{P}$ switching to have $x$ immediately above $y$ in $\mathbf{P}'$. Now let $\mathbf{Q}$ be obtained from $\mathbf{P}$ by only voter $j$ switching to have $x$ immediately above $y$. Hence in $\mathbf{Q}$, both $i$ and $j$ have $x$ immediately above $y$. As $\mathrm{dom}(F)$ is rich, $\mathbf{Q} \in \mathrm{dom}(F)$. Now observe that in the notation of Definition~\ref{PrefEq}, $\mathbf{Q}^i=\mathbf{P}'$ and $\mathbf{Q}^j=\mathbf{P}$, so Preferential Equality implies $F(\mathbf{P}')=F(\mathbf{P})$, as desired.\end{proof}

Another noteworthy fact is that Preferential Equality for individual voters is equivalent to Preferential Equality for coalitions of voters, as mentioned in Section~\ref{Intro}.

\begin{lemma}\label{PrefEq3} A voting rule $F$ on a rich domain of profiles satisfies Preferential Equality if and only if for any profile $\mathbf{P}$, $x,y\in X(\mathbf{P})$,  and $n\in \mathbb{N}_{>0}$, if there are $2n$ voters who rank $x$ immediately above $y$, then for any partition of those voters into two groups $I$ and $J$ of equal size, where $\mathbf{P}^I$ is obtained from $\mathbf{P}$ by all voters  in $I$ switching to ranking $y$ immediately above $x$, and $\mathbf{P}^J$ is obtained from $\mathbf{P}$ by all voters in $J$ switching to ranking $y$ immediately above $x$, we have  $F(\mathbf{P}^I)=F(\mathbf{P}^J)$.
\end{lemma}

\begin{proof} From right to left, the $n=1$ case is simply Preferential Equality as formulated in Definition~\ref{PrefEq}. From left to right, by $n$ applications of Preferential Compensation as in Lemma \ref{PrefEq2}, we transform $\mathbf{P}^I$ into $\mathbf{P}^J$, yielding $F(\mathbf{P}^I)=F(\mathbf{P}^J)$.
\end{proof}

Recall from Section \ref{Intro} that Instant Runoff Voting (IRV) violates Preferential Equality, at least hypothetically. In fact, this is far from hypothetical. Table~\ref{IRV_Violations} shows lower bounds\footnote{The reason these frequencies are only lower bounds is that we did not exhaustively search for coalitions $I$ and $J$ as described in the text but rather used a heuristic to find such coalitions $I$ and $J$.} on the frequencies in four election databases of profiles $\mathbf{P}$ that witness a violation by IRV of the coalitional version of Preferential Equality in the following sense: there exist candidates $x$ and $y$ and equally sized coalitions $I$ and $J$ of voters ranking $x$ immediately above $y$ such that \textit{if}, \textit{counterfactually}, the voters in $I$ were to switch to ranking $y$ immediately above $x$, this would have a different effect than if instead the voters in $J$ were to make exactly the same switch. Table~\ref{IRV_Violations} also gives lower bounds on the frequencies of profiles $\mathbf{P}$ that witness a violation by IRV of the obvious coalitional version of Preferential Compensation, which is equivalent to Preferential Compensation: there exists candidates $x$ and $y$ and equally sized coalitions $I$ and $J$ of voters ranking $x$ immediately above $y$ and $y$ immediately above $x$, respectively, such that each coalition flipping their ranking of $\{x,y\}$ changes the IRV outcome. We include frequencies conditional on there not being a candidate who receives a majority of first-place votes in the first round, known as an absolute majority winner (AMW), since if there is an AMW, then IRV, Plurality voting, and Condorcet voting rules all agree. The frequencies show that when IRV elections are interesting (i.e., there is no AMW), it is not unusual that there are coalitions with unequal power in terms of how they could have changed the election.

\begin{table}[h]
\begin{center}

\footnotesize
\begin{tabular}{lrrr}
\toprule
  & \# Relevant & Lower Bound on  & Lower Bound on \\
Dataset& Profiles& Frequency of PEV & Frequency of PCV \\
\midrule
Stable Voting Dataset & 216 & 0.35   &  0.22  \\
CIVS Dataset & 1,216  & 0.36   &  0.31    \\ 
PrefLib Politics (soi) & 308   &  0.03  & 0.05   \\
\citealt{Otis2022} dataset & 448  &  0.07   &   0.08    \\
\bottomrule
\toprule
    &\# Relevant   & Lower Bound on  & Lower Bound on \\
  & Profiles with  &Frequency of PEV & Frequency of PCV \\
Dataset &    No AMW &Conditional on No AMW &  Conditional on No AMW\\
\midrule
Stable Voting Dataset   &  86 & 0.70 &  0.48  \\
CIVS Dataset    & 590  & 0.66  &   0.57  \\ 
PrefLib Politics (soi)      & 104   & 0.09 & 0.16  \\
\citealt{Otis2022} dataset   & 194  &  0.17  & 0.19    \\
\bottomrule
\end{tabular}
\end{center}
\caption{The relevant profiles are obtained from elections in which (i) each voter submits a linear order on some subset  of the set of candidates (so we do not allow voters to rank candidates as tied, though we allow them to leave some candidates unranked, as in IRV elections), (ii) there is a unique IRV winner, and (iii) there are no ties when iteratively removing candidates with the fewest first-place votes; for each of these elections, we restrict the profile to the ``top three'' candidates (which allows the use of a simple heuristic to check for violations of Preferential Equality and Preferential Compensation), defined as the IRV winner plus the last two losers to be eliminated.  The last two columns contain lower bounds on the frequencies of relevant profiles witnessing a Preferential Equality Violation (PEV) or Preferential Compensation Violation (PCV) for IRV, where in the lower half of the table these frequencies are conditional on there not being an absolute majority winner (AMW). The elections come from the Stable Voting Dataset 1-28-2026 (\citealt{HP2026data}), CIVS Dataset 2024-12-15 (\citealt{CIVS2024,Myers2024}), PrefLib Politics (soi files as of 12-20-2024) (\citealt{MatteiWalsh2013}), and \citealt{Otis2022}.  Note that we restricted to elections that are not ``test'' elections in the CIVS Dataset.}\label{IRV_Violations}
\end{table}

The following is immediate from the definition of head-to-head rules (Definition~\ref{H2H}).

\begin{fact} Any head-to-head voting rule---and hence any margin-based voting rule---satisfies Preferential Compensation and hence Preferential Equality.
\end{fact}

\noindent On the other hand, there are (contrived) voting rules satisfying Preferential Compensation that are not head-to-head rules, such as the following.

\begin{fact} Let $F$ be a voting rule on the domain of linear profiles defined as follows: given a profile $\mathbf{P}$, if there are some $x,y\in X(\mathbf{P})$ such that some voter ranks $x$ immediately above $y$ while some other voter ranks $y$ immediately above $x$, then $F$ agrees with the Borda voting rule on $\mathbf{P}$; otherwise $F$ agrees with the Plurality voting rule on $\mathbf{P}$. Then $F$ satisfies Preferential Compensation and hence Preferential Equality, but it is not a head-to-head voting rule.
\end{fact}

\subsection{Neutral Reversal}

Next we formally define Saari's \citeyearpar{Saari2003} axiom of Neutral Reversal, which has recently been used in the axiomatization of one margin-based voting rule, namely Split Cycle (see \citealt[\S~4.2]{HP2021}, \citealt{Ding2022}).

\begin{definition} Given a set $X\subseteq\mathcal{X}$ of candidates, a \textit{reversal pair} on $X$ is a pair $(L,L^{-1})$ where $L$ is a linear order of $X$ and $L^{-1}$ is the reverse of $L$ (i.e., $L^{-1}=\{(y, x) \mid (x, y) \in L\}$). We may also call a pair of voters who submit a reversal pair of rankings a ``reversal pair of voters.''

A voting rule $F$ satisfies Neutral Reversal if for any  $\mathbf{P},\mathbf{P}'\in\mathrm{dom}(F)$, if $\mathbf{P}'$ is obtained from $\mathbf{P}$ by adding two voters with a reversal pair of rankings on $X(\mathbf{P})$, then $F(\mathbf{P})=F(\mathbf{P}')$.
\end{definition}

To see that Neutral Reversal and Preferential Equality are independent axioms, recall that the Positive/Negative voting rule (\citealt{Lapresta2010}, \citealt{Heckelman2020}) assigns a candidate $1$ point for every voter who ranks them uniquely first and $-1$ point for every voter who ranks them uniquely last; the candidates whose sum of points is greatest are the winners. The following is straightforward to verify.

\begin{fact}\label{NRFact}$\,$
\begin{enumerate}
\item The Pareto rule (in either its strict or weak version), a head-to-head voting rule, satisfies Preferential Equality but not Neutral Reversal.
\item The Positive/Negative voting rule satisfies Neutral Reversal but not Preferential Equality.
\end{enumerate}
\end{fact}

\subsection{Proof of characterization}

The two axioms just introduced exactly characterize margin-based voting rules on the domain of linear profiles.

\begin{restatable}{theorem}{LinChar}\label{LinChar} If $F$ is a voting rule on the domain of linear profiles, then $F$ is margin-based if and only if $F$ satisfies  Preferential Equality and Neutral Reversal.
\end{restatable}

\noindent Before proving this theorem, we need some preliminary definitions and results. The first is standard.

\begin{definition} Given a profile $\mathbf{P}$ and $i\in V(\mathbf{P})$, let $\mathbf{P}_i$ be the one-voter profile that assigns to $i$ the ranking $\mathbf{P}(i)$. Given two profiles $\mathbf{P}, \mathbf{Q}$ with $X(\mathbf{P})=X(\mathbf{Q})$ and  $V(\mathbf{P})\cap V(\mathbf{Q}) = \varnothing$, the  \textit{disjoint union} of $\mathbf{P}$ and $\mathbf{Q}$, denoted $\mathbf{P}+\mathbf{Q}$, assigns to each $i\in V(\mathbf{P})$ the ranking $\mathbf{P}(i)$ and to each $i\in V(\mathbf{Q})$ the ranking $\mathbf{Q}(i)$.\end{definition}

Next, we recall from \citealt{McGarvey1953} and \citealt{Debord1987} how to construct preference profiles realizing an abstract margin matrix.

\begin{definition} Given a nonempty finite set $X\subseteq\mathcal{X}$ of candidates, a \textit{margin matrix} over $X$ is a function $m:X^2\to \mathbb{Z}$ such that for all $x,y\in X$, $m(x,y)=-m(y,x)$. Given margin matrices $m$ and $m'$ over $X$, we define the margin matrix $m-m'$ by $m-m'(x,y)=m(x,y)-m'(x,y)$ for $x,y\in X$.
\end{definition}
\noindent Note that for any profile $\mathbf{P}$, the function $\mathcal{M}_\mathbf{P}$ defined in (\ref{MDef}) is a margin matrix.

Given a margin matrix $m$, we want to construct a profile realizing those margins. To do so, we will use special types of rankings, specified in the following definition, also used by McGarvey \citeyearpar{McGarvey1953}. First, let us fix a linear order of the set $\mathcal{X}$ of all possible candidates, which we call the \textit{alphabetic order}. To avoid overloading small natural numbers with different meanings, we also fix a doubleton set $\{\top, \bot\}$. 

\begin{definition} Let $X$ be a finite subset of $\mathcal{X}$. We define $X^{(2)} = \{(x, y) \in X^2 \mid x \not= y\}$. For simplicity, we denote the pair $(x, y)$ simply by $xy$. For any $xy \in X^{(2)}$, with $L$ being the alphabetical order of $X \setminus \{x, y\}$, let $G^\top_{xy}$ be the linear order $xyL$, and  let $G^\bot_{xy}$ be the linear order $L^{-1}xy$. These two linear orders are together called a \textit{McGarvey pair}.
\end{definition}

Using this notation, we can define what is essentially Debord's \citeyearpar{Debord1987} construction of profiles realizing margin matrices, generalizing McGarvey's \citeyearpar{McGarvey1953} construction of profiles realizing majority graphs. The idea for margin matrices consisting of even entries is simple: focusing on the positive entries of the margin matrix, we can uniquely decompose it as a sum of simple margin matrices with a single positive entry $2$ (and a corresponding negative entry $-2$ across the diagonal), and each of these simple margin matrices can be realized by a pair of ballots $G^\top_{xy}$ and $G^\bot_{xy}$ where $(x, y)$ is the coordinate of the entry $2$.
However, for the purposes of proving Theorem \ref{LinChar}---in particular in light of the absence of the axiom of anonymity in its statement---we must be able to control precisely the identity of voters in this construction; it is not enough to end up with the same anonymized profile. Thus, below we introduce a function $v$ that essentially sets up an infinite queue of voters for each type of ballot we might need. Then a Debord-style construction returns a unique profile $\mathbf{D}_{v, m}$ for each margin matrix $m$ of even parity, if we always take voters from the beginning of each queue. This $\mathbf{D}_{v, m}$ is essentially a normal form of $m$, and the main observation leading to Theorem \ref{LinChar} is that each profile $\mathbf{P}$ can be transformed into its normal form $\mathbf{D}_{v, \mathcal{M}_{\mathbf{P}}}$ using only the operations of adding or removing reversal pairs (using the voters in the queue given by $v$) and doing preferential compensation. The function $v$ will be chosen so that all voters involved are ``fresh''.

\begin{definition}\label{DeBordConstruction} Let $X$ be a nonempty finite subset of $\mathcal{X}$ and $v: X^{(2)} \times \{\top, \bot\} \times \mathbb{N}_{>0} \to\mathcal{V}$ an injective map. Under this $v$, for any $i$ in the range of $v$, when $v(ab, \star, k) = i$, we call $i$ \emph{a designated $ab$-voter} and also \emph{a designated $ab$-$\star$-voter}. 

If $m$ is a margin matrix over $X$ such that $m(x,y)$ is even for all $x,y\in X$ and $m$ is not the zero matrix, then we construct a profile $\mathbf{D}_{v, m}$ as follows: 
\begin{enumerate}
    \item $V(\mathbf{D}_{v, m}) = \bigcup_{xy \in X^{(2)}}\{v(xy, \star, k) \mid 1 \le k \le m(x, y)/2,\; \star\in\{\top, \bot\}\}$;
    \item for any $i \in V(\mathbf{D}_{v, m})$, where $(xy, \star, k)$ is $i$'s pre-image under $v$, let $\mathbf{D}_{v,m}(i) = G^\star_{xy}$.
\end{enumerate}
\end{definition}
One might worry about how to define $\mathbf{D}_{v, m}$ when $m$ is a zero matrix, since by our definition all profiles must have at least one voter. In our later proof, we will introduce an independently fixed profile in this case. Now the following is easy to see from the definition of $\mathbf{D}_{v,m}$.

\begin{proposition} For $X$, $v$, and $m$ as in Definition \ref{DeBordConstruction}, $\mathcal{M}_{\mathbf{D}_{v, m}} = m$.\end{proposition}

Thanks to our being careful about the identity of voters in the definition of $\mathbf{D}_{v,m}$, we can characterize $\mathbf{D}_{v,m}$ as the unique profile with certain properties.

\begin{lemma}\label{KeyLem} Let $X$ and $v$ be as in Definition \ref{DeBordConstruction}. Suppose $\mathbf{P}$ is a profile such that:
\begin{enumerate}
    \item\label{KeyLem1} every voter $i \in V(\mathbf{P})$ is in the range of $v$, and where $(xy, \star, k)$ is $i$'s pre-image under $v$, we have $\mathbf{P}(i) = G_{xy}^\star$;
    \item\label{KeyLem2} $v(xy, \top, k) \in V(\mathbf{P})$ iff $v(xy, \bot, k) \in V(\mathbf{P})$;
    \item\label{KeyLem3} if $v(xy, \star, k) \in V(\mathbf{P})$, then for any $1 \le k' \le k$, we have $v(xy, \star, k') \in V(\mathbf{P})$;
    \item\label{KeyLem4} if $v(xy, \star, k) \in V(\mathbf{P})$ for some $\star \in \{\top, \bot\}$ and $k \in \mathbb{N}_{>0}$, then $v(yx, \star', k') \not\in V(\mathbf{P})$ for every $\star' \in \{\top, \bot\}$ and $k' \in \mathbb{N}_{>0}$.
\end{enumerate}
Then $\mathbf{P} = \mathbf{D}_{v, \mathcal{M}_\mathbf{P}}$.
\end{lemma}
\begin{proof}
    Suppose $\mathbf{P}$ is a profile satisfying the four conditions. 
    By condition \ref{KeyLem1}, we only need to show that $V(\mathbf{P}) = V(\mathbf{D}_{v, \mathcal{M}_\mathbf{P}})$. 
    For every $xy \in X^{(2)}$, let $\mathbf{P}_{|xy}$ be the profile\footnote{This ``profile'' could be empty, but margins can be defined on the empty profile trivially.} that keeps only the designated $xy$-voters in $\mathbf{P}$. 
    Note that $\mathbf{P}$ is the disjoint union of $\mathbf{P}_{|xy}$ for all $xy \in X^{(2)}$.
    By conditions \ref{KeyLem1} and~\ref{KeyLem2}, for any $xy, x'y' \in X^{(2)}$ such that $\{x, y\}\neq \{x', y'\}$, we have $\mathcal{M}_{\mathbf{P}_{|x'y'}}(x, y) = 0$, since $G_{x'y'}^\top$ and $G_{x'y'}^{\bot}$ rank $x, y$ differently.
    This means $\mathcal{M}_{\mathbf{P}}(x, y) = \mathcal{M}_{\mathbf{P}_{|xy}}(x, y) - \mathcal{M}_{\mathbf{P}_{|yx}}(y, x) = |V(\mathbf{P}_{|xy})| - |V(\mathbf{P}_{|yx})|$ for any $xy \in X^{(2)}$. 
    By condition \ref{KeyLem4}, at most one of $\mathbf{P}_{|xy}$ and $\mathbf{P}_{|yx}$ can be non-empty. 
    This means that if $\mathcal{M}_{\mathbf{P}}(x, y) \le 0$, then $V(\mathbf{P}_{|xy})$ must be empty, so trivially we have $V(\mathbf{P}_{|xy}) = \{v(xy, \star, k) \mid 1 \le k \le \mathcal{M}_{\mathbf{P}}(x, y) / 2, \star \in \{\top, \bot\}\}$.
    On the other hand, if $\mathcal{M}_{\mathbf{P}}(x, y) > 0$, then $\mathbf{P}_{|yx}$ must be empty, 
    meaning that $|V(\mathbf{P}_{|xy})| = \mathcal{M}_{\mathbf{P}}(x, y)$. 
    By conditions \ref{KeyLem2} and \ref{KeyLem3}, this means $V(\mathbf{P}_{|xy})$ must be 
    $\{v(xy, \star, k) \mid 1 \le k \le \mathcal{M}_{\mathbf{P}}(x, y)/2, \star \in \{\top, \bot\}\}$. This completes the proof.
\end{proof}

We are now in a position to prove Theorem~\ref{LinChar}.

\LinChar*

\begin{proof} It is easy to see that if $F$ is margin-based, then $F$ satisfies Preferential Equality and Neutral Reversal. Conversely, suppose $F$ satisfies Preferential Equality and Neutral Reversal. To show that $F$ is margin-based, suppose $\mathcal{M}_\mathbf{P}=\mathcal{M}_\mathbf{Q}$ (call this $m$). We must show that $F(\mathbf{P})=F(\mathbf{Q})$.
Since $\mathbf{P}$ and $\mathbf{Q}$ are linear profiles with the same margins (and at least $2$ candidates), the parity of the number of voters is the same in $\mathbf{P}$ as in $\mathbf{Q}$, and $X(\mathbf{P}) = X(\mathbf{Q})$ (call it $X$). Fix an injective map \[v: X^{(2)}\times\{\top, \bot\}\times\mathbb{N}_{>0}\to \mathcal{V}\setminus (V(\mathbf{P})\cup V(\mathbf{Q})).\] Recall that if $i = v(xy, \star, k)$, $i$ is called a designated $xy$-voter and also a designated $xy$-$\star$-voter. Also, fix a profile with two voters not in $V(\mathbf{P}) \cup V(\mathbf{Q})$ voting a reversal pair of rankings. For notational convenience, call this profile $\mathbf{D}_{v, 0}$, where $0$ is the zero matrix of the same dimension as $m$.

Case 1: the parity of the number of voters is even. Pick a pair of voters $i,j\in V(\mathbf{P})$. Our goal is to transform their rankings $\mathbf{P}_i,\mathbf{P}_j$ into a reversal pair, by transforming $\mathbf{P}_i$ to be the reverse of $\mathbf{P}_j$. If, e.g., we want to flip some adjacent candidates $xy$ in $\mathbf{P}_i$ to $yx$, pick the smallest $k \in \mathbb{N}_{>0}$ such that neither $v(xy, \top, k)$ nor $v(xy, \bot, k)$ has been added to the profile and then add them with the rankings $yxL$ (not $xyL$) and $L^{-1}xy$, respectively, where $L$ is the alphabetic order of $X\setminus \{x,y\}$. Thus, we have added a pair of voters with a reversal pair of rankings, which by Neutral Reversal does not change the output of $F$. Next, flip the $xy$ in $\mathbf{P}_i$ to $yx$, while also changing the ranking of voter $v(xy, \top ,k)$ from $yxL$ to $xyL$, which by Preferential Compensation (recall Lemma \ref{PrefEq2}) does not change the output of $F$. Note that now voter $v(xy, \top, k)$'s ranking is $G^\top_{xy}$, and voter $v(xy, \bot, k)$'s ranking is $G^\bot_{xy}$. By continuing in this way, introducing sufficiently many McGarvey pairs, we transform $\mathbf{P}_i, \mathbf{P}_j$ into the desired reversal pair. Then by Neutral Reversal, we can delete voters $i,j$ from the profile without changing the output of $F$. See Example \ref{KeyEx} below for an illustration.

We do this for every pair of voters from $\mathbf{P}$, obtaining a profile $\mathbf{P}^+$ containing none of the original voters. In fact, it is easy to observe that $\mathbf{P}^+$ satisfies the first three conditions of Lemma \ref{KeyLem}, since all old voters in $\mathbf{P}$ have been deleted, and new voters were taken from the front of the each queue given by $v$ in pairs. It is also clear that $\mathcal{M}_{\mathbf{P}^+} = m$ since we are only using operations that maintain the margin matrix.

Next, for any $xy$ such that there are both designated $xy$-voters and also designated $yx$-voters in $V(\mathbf{P}^+)$, by the second condition of Lemma \ref{KeyLem} there must be designated $xy$-$\top$-, $xy$-$\bot$-, $yx$-$\top$-, and $yx$-$\bot$-voters in $V(\mathbf{P}^+)$. Then we remove for each of these four categories a voter whose pre-image under $v$ has the largest third coordinate. Note that removing these four voters amounts to deleting two reversal pairs, since $G^\top_{xy}$ and $G^\bot_{yx}$ form a reversal pair and $G^\bot_{xy}$ and $G^\top_{yx}$ form a reversal pair. Thus, by Neutral Reversal it does not change the output of $F$. We continue doing so until there are either no designated $xy$-voters or no designated $yx$-voters. The second condition in Lemma \ref{KeyLem} is preserved by the above operation of removing two pairs of voters, so we are guaranteed to hit one of these two cases. 
If all voters would be removed at the final step of this process, we first add the profile $\mathbf{D}_{v, 0}$ fixed at the beginning before the final removal, which leaves the output of $F$ unchanged.

Call the resulting profile $\mathbf{P}^\dagger$. It is now clear that $\mathcal{M}_{\mathbf{P}^\dagger}$ is still $m$. If the final addition of a reversal pair did not happen, then $\mathbf{P}^\dagger$ satisfies all the four conditions of Lemma \ref{KeyLem}, and thus, $\mathbf{P}^\dagger = \mathbf{D}_{v, m}$. If the final addition did happen, then $m$ must be the zero matrix since performing only margin-preserving operations led us to the empty ``profile'', and by stipulation $\mathbf{P}^\dagger = \mathbf{D}_{v, m}$. Also, as we have only used Neutral Reversal and Preferential Compensation, we have $F(\mathbf{P}^\dagger) = F(\mathbf{P}^+) = F(\mathbf{P})$.

Applying the same process to $\mathbf{Q}$, we obtain $\mathbf{Q}^\dagger$ such that $F(\mathbf{Q}^\dagger)=F(\mathbf{Q})$ and $\mathbf{Q}^\dagger = \mathbf{D}_{v, m}$. Then $F(\mathbf{P}) = F(\mathbf{P}^\dagger) = F(\mathbf{D}_{v, m}) = F(\mathbf{Q}^\dagger) = F(\mathbf{Q})$, as desired.

Case 2: the parity of the number of voters in $\mathbf{P}$ and $\mathbf{Q}$ is odd. By Neutral Reversal, we may assume without loss of generality that some voter $i$ appears in both $\mathbf{P}$ and $\mathbf{Q}$ with the same ballot, so $\mathbf{P}_i=\mathbf{Q}_i$. Now for the part of the profile $\mathbf{P}$ that excludes the voter $i$, proceed exactly as in the previous case, including the final addition of the fixed reversal pair if $\mathcal{M}_\mathbf{P}-\mathcal{M}_{\mathbf{P}_i}$ is the zero matrix. Then after adding the voter $i$ back, the resulting profile $\mathbf{P}^\dagger$ is identical to  $\mathbf{P}_i + \mathbf{D}_{v, \mathcal{M}_\mathbf{P}-\mathcal{M}_{\mathbf{P}_i}}$ and hence 
\[
F(\mathbf{P})=F(\mathbf{P}_i+ \mathbf{D}_{v, \mathcal{M}_\mathbf{P}-\mathcal{M}_{\mathbf{P}_i}}).
\] 
Do the same for $\mathbf{Q}$, so $\mathbf{Q}^\dagger$ is identical to $\mathbf{Q}_i+ \mathbf{D}_{v, \mathcal{M}_\mathbf{Q}-\mathcal{M}_{\mathbf{Q}_i}}$ and hence 
\[
F(\mathbf{Q})= F(\mathbf{Q}_i+ \mathbf{D}_{v, \mathcal{M}_\mathbf{Q}-\mathcal{M}_{\mathbf{Q}_i}}).
\] 
Then since  $\mathbf{P}_i=\mathbf{Q}_i$ and $\mathcal{M}_\mathbf{P}=\mathcal{M}_\mathbf{Q}$, we have $\mathcal{M}_{\mathbf{P}} - \mathcal{M}_{\mathbf{P}_i} = \mathcal{M}_{\mathbf{Q}} - \mathcal{M}_{\mathbf{Q}_i}$. This means 
\[
\mathbf{P}_i+ \mathbf{D}_{v, \mathcal{M}_\mathbf{P}-\mathcal{M}_{\mathbf{P}_i}} = 
\mathbf{Q}_i+ \mathbf{D}_{v, \mathcal{M}_\mathbf{Q}-\mathcal{M}_{\mathbf{Q}_i}},
\] 
so $F(\mathbf{P})=F(\mathbf{Q})$.\end{proof}

Below is an example showcasing how to reduce a profile $\mathbf{P}$ to its normal form $\mathbf{D}_{v, \mathcal{M}_{\mathbf{P}}}$ by adding and removing reversal pairs and applying preferential compensation.
\begin{example}\label{KeyEx}
Consider the following profiles $\mathbf{P}$ and $\mathbf{Q}$, which realize the same margin graph (recall Section \ref{Intro}):
\begin{center}
\begin{minipage}{1.9in}
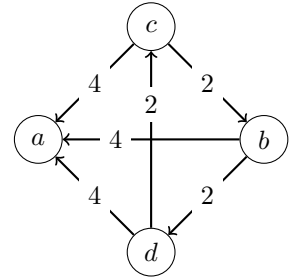
\begin{center}
\begin{tabular}{ccccc}
$i_1,i_2$ & $i_3$ & $i_4$ & $i_5$ & $i_6$  \\\hline 
$d$ & \textcolor{blue}{$c$} & \textcolor{blue}{$b$} & $b$ & $c$  \\ 
$c$ & \textcolor{blue}{$a$} & \textcolor{blue}{$d$} & $d$ & $b$  \\ 
$b$ & \textcolor{blue}{$b$} & \textcolor{blue}{$a$} & $c$ & $d$  \\ 
$a$ & \textcolor{blue}{$d$} & \textcolor{blue}{$c$} & $a$ & $a$ 
\end{tabular}
\vspace{.1in}
\textbf{P}
\end{center}
\end{minipage}\begin{minipage}{1.9in}
\begin{center}
\begin{tabular}{cccc}
$j_1$ & $j_2,j_3$ & $j_4$ & $j_5,j_6$\\\hline 
$a$ & $b$ & $c$ & $d$\\ 
$c$ & $d$ & $b$ & $c$\\ 
$b$ & $c$ & $d$ & $b$\\ 
$d$ & $a$ & $a$ & $a$
\end{tabular}
\vspace{.1in}
\textbf{Q}
\end{center}
\end{minipage}\begin{minipage}{1.9in}
\begin{center}
\begin{tikzpicture}
\node[circle,draw,minimum width=0.25in] at (0,0)      (a) {$a$}; 
\node[circle,draw,minimum width=0.25in] at (3,0)      (b) {$b$}; 
\node[circle,draw,minimum width=0.25in] at (1.5,1.5)  (c) {$c$}; 
\node[circle,draw,minimum width=0.25in] at (1.5,-1.5) (d) {$d$};
\path[->,draw,thick] (b) to[pos=.7] node[fill=white] {$4$} (a);
\path[->,draw,thick] (c) to node[fill=white] {$4$} (a);
\path[->,draw,thick] (d) to node[fill=white] {$4$} (a);
\path[->,draw,thick] (c) to node[fill=white] {$2$} (b);
\path[->,draw,thick] (b) to node[fill=white] {$2$} (d);
\path[->,draw,thick] (d) to[pos=.7]  node[fill=white] {$2$} (c);
\end{tikzpicture}
\end{center}\end{minipage}

\end{center}

\noindent Now let $P$ and $P'$ be the rankings of $i_3$ and $i_4$, respectively, in $\mathbf{P}$. Note that $P$ is almost the reverse of $P'$; it would be the reverse of $P'$ if only $ca\textcolor{red}{bd}$ were instead $ca\textcolor{red}{db}$. To flip $bd$ in $P$ to $db$, following the strategy in the proof of Theorem \ref{LinChar}, we first add the reversal pair $(dbL, L^{-1}bd)$, where $L=ac$, obtaining
\begin{center}
\begin{tabular}{ccccccc}
$i_1,i_2$ & $i_3$ & $i_4$ & $i_5$ & $i_6$ & $v(bd,\top,1)$ & $v(bd,\bot,1)$  \\\hline 
$d$ & \textcolor{blue}{$c$} & \textcolor{blue}{$b$} & $b$ & $c$ & \textcolor{medgreen}{$d$} & \textcolor{medgreen}{$c$}  \\ 
$c$ & \textcolor{blue}{$a$} & \textcolor{blue}{$d$} & $d$ & $b$ & \textcolor{medgreen}{$b$} & \textcolor{medgreen}{$a$}  \\ 
$b$ & \textcolor{blue}{$b$} & \textcolor{blue}{$a$} & $c$ & $d$ & \textcolor{medgreen}{$a$} & \textcolor{medgreen}{$b$} \\ 
$a$ & \textcolor{blue}{$d$} & \textcolor{blue}{$c$} & $a$ & $a$ & \textcolor{medgreen}{$c$} & \textcolor{medgreen}{$d$}
\end{tabular}

\vspace{.1in}
$\mathbf{P}'$
\end{center}
and then flip $bd$ to $db$ in $P$ while flipping $dbL$ to $bdL$, obtaining
\begin{center}
\begin{tabular}{ccccccc}
$i_1,i_2$ & $i_3$ & $i_4$ & $i_5$ & $i_6$ & $v(bd,\top,1)$ & $v(bd,\bot,1)$  \\\hline 
$d$ & \textcolor{blue}{$c$} & \textcolor{blue}{$b$} & $b$ & $c$ & \textcolor{red}{$b$} & \textcolor{medgreen}{$c$}  \\ 
$c$ & \textcolor{blue}{$a$} & \textcolor{blue}{$d$} & $d$ & $b$ & \textcolor{red}{$d$} & \textcolor{medgreen}{$a$}  \\ 
$b$ & \textcolor{red}{$d$} & \textcolor{blue}{$a$} & $c$ & $d$ & \textcolor{medgreen}{$a$} & \textcolor{medgreen}{$b$}  \\ 
$a$ & \textcolor{red}{$b$} & \textcolor{blue}{$c$} & $a$ & $a$ & \textcolor{medgreen}{$c$} & \textcolor{medgreen}{$d$}
\end{tabular}

\vspace{.1in}
$\mathbf{P}''$
\end{center}
Now we can delete the reversal pair of $i_3$ and $i_4$, resulting in 
\begin{center}
\begin{tabular}{ccccc}
$i_1,i_2$  & $i_5$ & $i_6$ &$v(bd,\top,1)$ & $v(bd,\bot,1)$  \\\hline 
$d$   & $b$ & $c$ & \textcolor{red}{$b$} & \textcolor{medgreen}{$c$}  \\ 
$c$   & $d$ & $b$ & \textcolor{red}{$d$} & \textcolor{medgreen}{$a$}  \\ 
$b$ &   $c$ & $d$ & \textcolor{medgreen}{$a$} & \textcolor{medgreen}{$b$}  \\ 
$a$ &   $a$ & $a$ & \textcolor{medgreen}{$c$} & \textcolor{medgreen}{$d$}
\end{tabular}

\vspace{.1in}
$\mathbf{P}'''$
\end{center}
Thus, we have replaced $ P,P'$ with a McGarvey pair. In this case, the newly added voters end up having exactly the same ballots as the deleted $i_3$ and $i_4$. But this is necessary if we want to convert two profiles sharing the same margin matrix to exactly the same profile even using the same voters. In other cases, more steps are needed to transform an initial pair into a reversal pair, which can then be eliminated. For example, to turn $i_5$'s ballot into the reversal of $i_6$'s, we need to first swap $a$ with $c$ and then $d$ and then $b$, and then swap $d$ with $b$. These steps introduce $v(ca, \top, 1)$ and $v(ca, \bot, 1)$, $v(da, \top, 1)$ and $v(da, \bot, 1)$, $v(ba, \top, 1)$ and $v(ba, \bot, 1)$, and $v(bd, \top, 2)$ and $v(bd, \bot, 2)$ (since $v(bd, \top, 1)$ and $(bd, \bot, 1)$ have already been used). The process of adding McGarvey pairs may also create some quadruples of the form $G^\top_{xy}, G^\bot_{xy}, G^\top_{yx}, G^\bot_{yx}$, which can be eliminated at the end. Ultimately we obtain the following profile, $\mathbf{D}_{v, \mathcal{M}_\mathbf{P}}$, where the number $k$ above a ranking $L$ indicates the number of voters---whose identities are determined by the function $v$---submitting the ranking~$L$:
\begin{center}
\begin{tabular}{cccccccccccc}
$2$ & $2$ & $1$ & $1$ & $2$ & $2$ & $1$ & $1$ & $2$ & $2$ & $1$ & $1$\\\hline 
\textcolor{purple}{$b$} & $d$ & \textcolor{orange}{$b$} & $c$ & \textcolor{purple}{$c$} & $d$ & \textcolor{orange}{$c$} & $d$ & \textcolor{purple}{$d$} & $c$ & \textcolor{orange}{$d$} & $b$\\ 
\textcolor{purple}{$a$} & $c$ & \textcolor{orange}{$d$} & $a$ & \textcolor{purple}{$a$} & $b$ & \textcolor{orange}{$b$} & $a$ & \textcolor{purple}{$a$} & $b$ & \textcolor{orange}{$c$} & $a$\\ 
$c$ & \textcolor{purple}{$b$} & $a$ & \textcolor{orange}{$b$} & $b$ & \textcolor{purple}{$c$} & $a$ & \textcolor{orange}{$c$} & $b$ & \textcolor{purple}{$d$} & $a$ & \textcolor{orange}{$d$}\\ 
$d$ & \textcolor{purple}{$a$} & $c$ & \textcolor{orange}{$d$} & $d$ & \textcolor{purple}{$a$} & $d$ & \textcolor{orange}{$b$} & $c$ & \textcolor{purple}{$a$} & $b$ & \textcolor{orange}{$c$}
\end{tabular}
\end{center}

\noindent The same process starting from $\mathbf{Q}$ leads to $\mathbf{D}_{v,\mathcal{M}_\mathbf{Q}}=\mathbf{D}_{v,\mathcal{M}_\mathbf{P}}$.
\end{example}

\subsection{Weakening Neutral Reversal to Block Invariance}

 A natural question is whether we can modify Theorem \ref{LinChar} by weakening Neutral Reversal to the following axiom of Block Invariance (\citealt{HP2025}; cf.~\citealt[Def.~4]{Merlin2003}).

 \begin{definition} A voting rule $F$ satisfies \textit{Block Invariance} if for any $\mathbf{P},\mathbf{P}'\in \mathrm{dom}(F)$ with $X(\mathbf{P})=X(\mathbf{P}')$, if $\mathbf{P}'$ is obtained from $\mathbf{P}$ by adding, for each linear order $L$ of  $X(\mathbf{P})$, exactly one voter submitting $L$ as their ranking, then $F(\mathbf{P})=F(\mathbf{P}')$.

 \end{definition}
 
\noindent Unlike Neutral Reversal, Block Invariance is satisfied by the Plurality voting rule and IRV rule. However, the answer is that we cannot weaken Theorem \ref{LinChar} by replacing Neutral Reversal with Block Invariance, as shown by the following example.
 
 \begin{example} Fix some sequence of margin-based voting rules $F_0,F_1,\dots$. Then define a new voting rule $F$ as follows: if $|V(\mathbf{P})|$ is congruent to $k$ modulo $|X(\mathbf{P})|!$, for $k\in \{0,\dots,|X(\mathbf{P})|!-1\}$, then $F(\mathbf{P})=F_k(\mathbf{P})$. Since each $F_k$ is margin-based, $F$ clearly satisfies Preferential Equality, whose definition does not involve a change in the number of voters. Moreover, adding to $\mathbf{P}$ a block of all linear orders of $X(\mathbf{P})$ results in a profile that is in the same congruence class modulo $|X(\mathbf{P})|!$ as $\mathbf{P}$ is and has the same margins, so $F$ satisfies Block Invariance. But it is easy to choose the sequence $F_0,F_1,\dots$ so that $F$ is not margin-based, because whether the number of voters in a profile $\mathbf{P}$ falls into a particular congruence class modulo $|X(\mathbf{P})|!$ is not uniquely determined by $\mathcal{M}_\mathbf{P}$.\end{example}

Despite this example, if we avail ourselves of the normatively highly compelling axiom of \emph{Homogeneity} (\citealt{Smith1973}), we can indeed weaken Neutral Reversal to Block Invariance. For the definition of Homogeneity, fix a linear order of $\mathcal{V}$, which we call the \textit{alphabetic order} of $\mathcal{V}$.

\begin{definition} Given a profile $\mathbf{P}$ and $n \ge 1$, let $n\mathbf{P}=\mathbf{P}+\mathbf{P}'$ where $\mathbf{P}'$ is $n-1$ copies of $\mathbf{P}$ using the first $(n-1)|V(\mathbf{P})|$ voters from $\mathcal{V}\setminus V(\mathbf{P})$ according to the alphabetic order of $\mathcal{V}$. A voting rule satisfies \emph{Homogeneity} and is called \textit{homogeneous} if for any profile $\mathbf{P}\in\mathrm{dom}(F)$, we have $F(\mathbf{P})=F(n\mathbf{P})$ for all $n \ge 2$.
\end{definition}

 \noindent One of the only voting rules discussed in the literature that violates Homogeneity is the Dodgson rule (see \citealt{Brandt2009}), but this is considered a pathology of the rule. It would be hard to justify to voters that if every voter brought a twin to the election who submitted the same ballot, this could change the election outcome.

To see that not \textit{all} margin-based voting rules satisfy Homogeneity, we can cook up an artificial example like the following, which is easily verified.

\begin{fact}\label{Inhomogeneous} Consider the voting rule $F$ defined as follows: if all margins in the profile $\mathbf{P}$ are below some threshold $\lambda$, then use some margin-based voting rule, e.g., Beat Path (\citealt{Schulze2011}); otherwise use some other such rule, e.g., Ranked Pairs (\citealt{Tideman1987}). Then $F$ is margin-based but does not satisfy Homogeneity.\end{fact}

Let us now prove that by assuming Homogeneity, we can weaken Neutral Reversal to Block Invariance in the characterization of margin-based rules. The proof is a simple modification of the proof using Neutral Reversal.

\begin{theorem}\label{BlockInvarianceThm}
  Let $F$ be a voting rule satisfying Homogeneity. Then $F$ is margin-based if and only if it satisfies Preferential Equality and Block Invariance.
\end{theorem}
\begin{proof}
  For the non-trivial direction, suppose we have profiles $\mathbf{P}$ and $\mathbf{Q}$ such that $\mathcal{M}_\mathbf{P} = \mathcal{M}_\mathbf{Q}$.
  Let $n = |X(\mathbf{P})| = |X(\mathbf{Q})|$, $\mathbf{P}' = n!\mathbf{P}$, and $\mathbf{Q}' = n!\mathbf{Q}$.
  Then again $\mathcal{M}_{\mathbf{P}'} = \mathcal{M}_{\mathbf{Q}'}$, but more importantly, both $|V(\mathbf{P}')|$ and $|V(\mathbf{Q}')|$ are multiples of $n!$. By Homogeneity, $F(\mathbf{P})=F(\mathbf{Q})$ iff $F(\mathbf{P}')=F(\mathbf{Q}')$, so it suffices to prove the latter to show that $F$ is margin-based.

  Our first step is to make the voters in $\mathbf{P}'$ and $\mathbf{Q}'$ vote blocks of all possible linear orders of the candidates and then use Block Invariance to eliminate them.
  In our proof of Theorem \ref{LinChar}, we used Neutral Reversal for this step.
  To use only Block Invariance, when introducing the two new designated $xy$-voters with the reversal rankings $yxL$ and $L^{-1}xy$, we also introduce new voters for all other linear orders of the candidates, so that a new block is introduced.
  Call this block a \emph{designated-$xy$-block}.
  To formally track the identity of these added voters, we again use a $v$ function, now of the type $X^{(2)} \times \mathsf{Lin}(X) \times \mathbb{N}_{>0} \to \mathcal{V} \setminus (V(\mathbf{P}') \cup V(\mathbf{Q}'))$, where $\mathsf{Lin}(X)$ is the set of all linear orders on $X$.
  That is, we are replacing $\{\top, \bot\}$ with $\mathsf{Lin}(X)$.
  A designated-$xy$-block is then a set $\{v(xy, R, k) \mid R \in \mathsf{Lin}(X)\}$ for some $k$ (the $k$-th designated-$xy$-block).
  The voter $v(xy, R, k)$ will always vote $R$ for any $R \in \mathsf{Lin}(X)$ except when $R = yxL$, where again $L$ is the alphabetical order on $X \setminus \{x, y\}$.
  When introducing such a block without changing the outcome of the election, we let $v(xy, yxL, k)$ vote $yxL$.
  Then, after a use of Preferential Compensation (recall Lemma \ref{PrefEq2}) to swap $xy$ to $yx$ in the ballot of an old voter, the voter $v(xy, yxL, k)$ votes $xyL$, and now these $n!$ voters in $\{v(xy, R, k) \mid R \in \mathsf{Lin}(X)\}$ form a profile whose margin matrix is equal to that of the profile with just the McGarvey pair $xyL$ and $L^{-1}xy$, i.e., with a margin of $2$ for $xy$, $-2$ for $yx$, and $0$ for all other pairs.
  After all the old voters are eliminated, we then have to make sure that no designated-$xy$-blocks and designated-$yx$-blocks exist simultaneously.
  Observe that if we have both a designated-$xy$-block and a designated-$yx$-block, then they in fact form two blocks of all linear orders of the candidates, since $v(xy, yxL, k)$ and $v(yx, xyL, k')$ now vote each other's original ballot, namely $xyL$ and $yxL$, from before we used Preferential Compensation.
  Thus, they can be removed together without changing the outcome by Block Invariance.
  If all voters are removed, we then add back a fixed block of all linear orders not using any of the original voters in either $\mathbf{P}$ and $\mathbf{Q}$.
  In exactly the same way as we proved Lemma \ref{KeyLem}, we can prove that there is only one profile to realize a given non-zero margin matrix $m$ of even parity using these designated blocks in the order given by the third index $k$ without redundancy.
  Thus, $\mathbf{P}'$ and $\mathbf{Q}'$ can be turned into the same profile using only Preferential Compensation and Block Invariance, so $F(\mathbf{P}')=F(\mathbf{Q}')$, as desired.
\end{proof}

\section{Characterization for all profiles, modulo homogeneity}\label{StrictWeakProfs}

In this section, we consider the domain of all profiles, as well as a domain intermediate between linear profiles and all profiles. Say that a strict weak order $P$ on $X$ is a \textit{linear order with bottom indifference} if whenever candidates $x$ and $y$ are ranked in a tie (i.e., neither $(x,y)\in P$ nor $(y,x)\in P$), then $x$ and $y$ are not ranked strictly above any candidates. Note that every linear order is considered a linear order with bottom indifference, even though it contains no indifference. Note also that the strict weak order where all candidates are tied together (represented by the empty set) is a linear order with bottom indifference as well. \begin{definition} Let $\mathsf{LOBI}$ be the domain of profiles in which each voter is assigned a linear order with bottom indifference.
\end{definition}
\noindent The importance of this domain is that in real political elections with ranked ballots, voters are not required to rank all the candidates. To turn such an election into one of our profiles, it is standard to consider all candidates unranked by a voter as being in a big tie below all candidates ranked by the voter. Thus, real political elections yield profiles in the $\mathsf{LOBI}$ domain. Moreover, some voting rules used in real political elections, such as IRV, are traditionally defined only for the $\mathsf{LOBI}$ domain (though generalizations have been considered, e.g., in \citealt{Delemazure2024}). Finally, as we will see (in Theorem \ref{StrictWeakChar}), the axioms sufficient to force the use of a margin-based voting rule when the domain is assumed to be the domain of all profiles might not be sufficient to force the use of a margin-based voting rule when the domain is assumed to be the $\mathsf{LOBI}$ domain of real political elections, so for normative purposes, we must analyze the $\mathsf{LOBI}$ case separately.

There is another possible but less common interpretation of unranked candidates, namely that they are \textit{noncomparable} with all ranked candidates. In the Appendix, we give characterizations of margin-based voting rules on a larger domain of ``profiles'' (now using a different definition than in Section~\ref{Intro}) in which voters are allowed to have noncomparabilities and hence are not required to submit strict weak orders.

\subsection{Tiebreaking Compensation}

Now that we are allowing ties in rankings (even if only at the bottom in $\mathsf{LOBI}$), we introduce an axiom motivated by the idea that each voter should have equal tiebreaking power.

\begin{definition} 
  An \emph{indifference class} in a strict weak order $P$ on $X$ is a nonempty set $Y \subseteq X$ such that for any $x, y \in Y$, $(x, y) \not\in P$, while for any $z \in X \setminus Y$ and $y \in Y$, $(z, y) \in P$ or $(y, z) \in P$. It is well known that indifference classes are the equivalence classes of the indifference equivalence relation $\sim_P$ defined by $x \sim_P y$ iff neither $(x, y)$ nor $(y, x)$ is in $P$. By a \emph{tie}, we mean a non-singleton indifference class.
      
  A voting rule $F$ satisfies Tiebreaking Compensation if for all profiles $\mathbf{Q},\mathbf{Q}'\in\mathrm{dom}(F)$, if in $\mathbf{Q}$ there are two voters $i$ and $j$ and a set $Y$ of candidates which is an indifference class for both $i$ and $j$, and $\mathbf{Q}'$ is obtained from $\mathbf{Q}$ by $i$ replacing the tie $Y$ by some linear order $L$ of $Y$ and $j$ replacing the tie $Y$ by $L^{-1}$, then $F(\mathbf{Q})=F(\mathbf{Q}')$.
\end{definition}

\noindent Thus, two voters breaking a tie between the same set of candidates but in opposite ways does not change the outcome, reflecting the idea that their opposition in tiebreaking balances out.

For voting rules on the domain of all profiles, Tiebreaking Compensation implies our earlier Preferential Compensation and hence Preferential Equality by Lemma~\ref{PrefEq2}.

\begin{lemma}\label{CompLem} If a voting rule on the domain of all profiles satisfies Tiebreaking Compensation, then it satisfies Preferential Compensation and hence Preferential Equality.
\end{lemma}
\begin{proof}  Let $\mathbf{P},\mathbf{P}'$ be as in the definition of Preferential Compensation in Lemma \ref{PrefEq2}, so $\mathbf{P}'$ is obtained from $\mathbf{P}$ by one voter switching $xy$ to $yx$ while another switches $yx$ to $xy$. Let $\mathbf{Q}$ be the profile exactly like $\mathbf{P}$ except that both voters rank $\{x, y\}$ in a tie (and the only tie since $\mathbf{P}$ has only linear ballots). Then setting $\mathbf{Q}':=\mathbf{P}$, we obtain $F(\mathbf{Q})=F(\mathbf{P})$ by Tiebreaking Compensation with $Y=\{x,y\}$. Similarly, by setting $\mathbf{Q}':=\mathbf{P}'$, we obtain $F(\mathbf{Q})=F(\mathbf{P}')$ by Tiebreaking Compensation. Hence  $F(\mathbf{P})=F(\mathbf{P}')$, as desired. Then $F$ also satisfies Preferential Equality by Lemma \ref{PrefEq2}.
\end{proof}
\noindent Note that since the profile $\mathbf{Q}$ used in the proof might not belong to the domain $\mathsf{LOBI}$, the same argument does not work for voting rules whose domain is $\mathsf{LOBI}$. Indeed, the implication does not hold for $\mathsf{LOBI}$, as shown by the following.

\begin{fact} Let $F$ be the voting rule defined on $\mathsf{LOBI}$ as follows: if the profile contains a tie or two voters who rank different candidates in last place, then $F$ outputs all candidates; otherwise $F$ outputs the Plurality winners. Then $F$ satisfies Tiebreaking Compensation\footnote{In the definition of Tiebreaking Compensation, $\mathbf{Q}$ has at least a tie, and since we assume here that $\mathbf{Q} \in \mathsf{LOBI}$, $\mathbf{Q}'$ contains two voters who rank different candidates in last place, so $F(\mathbf{Q})=X(\mathbf{Q})=X(\mathbf{Q}')=F(\mathbf{Q}')$.} but not Preferential Compensation or Preferential Equality (since Plurality does not satisfy these axioms).\end{fact}

Next observe that Tiebreaking Compensation, like Preferential Compensation, is independent of Neutral Reversal. The following is easy to verify.

\begin{fact}$\,$
\begin{enumerate}
\item The strict Pareto rule on the domain of all profiles satisfies Tiebreaking Compensation (and Preferential Compensation) but not Neutral Reversal.
\item The version of Minimax based on ``winning votes'' satisfies Neutral Reversal and Preferential Compensation but not Tiebreaking Compensation (e.g., taking two of the voters in Figure \ref{MarginsVsWinningVotes} who are indifferent between $a$ and $c$ and having them break the tie in opposite ways changes the set of winners from $\{a\}$ to $\{a,c\}$). 
\end{enumerate}
\end{fact}

\subsection{Neutral Indifference}\label{NeutInd}

The second axiom simply says that adding a voter with a fully indifferent ranking does not change the election outcome.  Technically the fully indifferent strict weak order is the empty relation $\varnothing$. We will also call such a ranking an ``empty ballot.''

\begin{definition} A voting rule $F$ satisfies Neutral Indifference if for any $\mathbf{P},\mathbf{P}'\in\mathrm{dom}(F)$, if $\mathbf{P}'$ is obtained from $\mathbf{P}$ by adding one voter who submits a fully indifferent ranking, then $F(\mathbf{P})=F(\mathbf{P}')$.
\end{definition}

\noindent This axiom captures the distinction between C2 rules and arbitrary head-to-head rules.

\begin{proposition}\label{C2Char} A head-to-head voting rule on the domain of all profiles (resp.~on $\mathsf{LOBI}$) is C2 if and only if it satisfies Neutral Indifference.
\end{proposition}

\begin{proof} The left-to-right direction is immediate. From right to left, assume a head-to-head voting rule $F$ satisfies Neutral Indifference, and consider profiles $\mathbf{P},\mathbf{P}'$ such that $\#_\mathbf{P}=\#_{\mathbf{P}'}$. If in addition $|V(\mathbf{P})|=|V(\mathbf{P}')|$, then $\mathcal{H}(\mathbf{P}) =\mathcal{H}(\mathbf{P}')$, so by the head-to-head assumption, $F(\mathbf{P})=F(\mathbf{P}')$. If one profile, say $\mathbf{P}$, has fewer voters, then add voters with fully indifferent ballots until we obtain a profile $\mathbf{P}^+$ with the same number of voters as $\mathbf{P}'$. Then by Neutral Indifference, $F(\mathbf{P})=F(\mathbf{P}^+)$, and by the head-to-head assumption, $F(\mathbf{P}^+)=F(\mathbf{P}')$, so $F(\mathbf{P})=F(\mathbf{P}')$, which shows that $F$ is~C2.\end{proof}

Note that the strict Pareto rule violates Neutral Indifference, while the version of Minimax based on winning votes satisfies it.

\subsection{Proof of characterization}

We are now ready to state and prove our characterizations for the additional domains.

\begin{theorem}\label{StrictWeakChar} Let $F$ be a voting rule satisfying Homogeneity.\footnote{In fact, for this theorem, it suffices to assume the weaker form of Homogeneity that only requires $F(\mathbf{P})=F(2\mathbf{P})$.}
\begin{enumerate}
    \item\label{StrictWeakChar2} If the domain of $F$ is $\mathsf{LOBI}$, then $F$ is margin-based if and only if $F$ satisfies Preferential Equality, Tiebreaking Compensation, and Neutral Indifference.
    \item\label{StrictWeakChar1} If the domain of $F$ is the domain of all profiles, then $F$ is margin-based if and only if $F$ satisfies Tiebreaking Compensation and Neutral Indifference.
\end{enumerate}
\end{theorem}
\begin{proof} If a profile contains a ranking with a tie, duplicate the profile, so we have two copies of the ranking, and then linearize the previously tied part in opposite ways. By Homogeneity and Tiebreaking Compensation, this does not change the output of $F$. In this way, we can change any two initial profiles into linear profiles. Then we reason exactly as in the proof of Theorem~\ref{LinChar} using Preferential Equality, which is available in part~\ref{StrictWeakChar1} thanks to Lemma~\ref{CompLem}, and Neutral Reversal, which is a consequence of Neutral Indifference (applied twice) together with Tiebreaking Compensation.\end{proof}

\begin{remark}\label{Replacements} In part \ref{StrictWeakChar2} of Theorem \ref{StrictWeakChar}, we can replace Tiebreaking Compensation by what could be called Pure Tiebreaking Compensation: if two voters have the \textit{same} ranking containing a tie $Y$, then changing the profile so that each voter linearizes $Y$ in opposite ways does not change the output of the voting rule. For in the proof of Theorem \ref{StrictWeakChar}.\ref{StrictWeakChar2}, we apply Tiebreaking Compensation to two copies of the same ranking. We could also replace Tiebreaking Compensation with Pure Tiebreaking Compensation in part \ref{StrictWeakChar1} of Theorem \ref{StrictWeakChar}, provided we add Preferential Equality, since unlike Tiebreaking Compensation, Pure Tiebreaking Compensation does not imply Preferential Equality for voting rules on the domain of all profiles.\footnote{Given $Y\subsetneq X$ and rankings $P$ and $P'$ of $X$, say that $P$ and $P'$ \textit{agree outside $Y$}, denoted $P\sim_Y P'$, if (a)~for all $x\in X\setminus Y$, either $xPy$ for all $y\in Y$ or $yPx$ for all $y\in Y$, and either $xP'y$ for all $y\in Y$ or $yP'x$ for all $y\in Y$,  and (b)~$P\setminus Y^2= P'\setminus Y^2$. Let $\mathbb{P}$ be the set of profiles $\mathbf{P}$ such that there are two voters whose rankings $P$ and $P'$ agree outside some $Y\subsetneq X(\mathbf{P})$ or are such that $P'=P^{-1}$. Let $F$ be the voting rule such that if $\mathbf{P}\in\mathbb{P}$, then $F(\mathbf{P})=X(\mathbf{P})$, and otherwise $F(\mathbf{P})$ is the set of Plurality winners. Observe that if $\mathbf{Q}'$ is obtained from $\mathbf{Q}$ by an application of Pure Tiebreaking Compensation, then $\mathbf{Q},\mathbf{Q}'\in\mathbb{P}$, so $F$ satisfies Pure Tiebreaking Compensation. However, $F$ violates Preferential Equality for the profile $\mathbf{P}$ with just two ballots, $abcde$ and $dabec$, since switching $ab$ to $ba$ in the first ballot vs.~the second ballot leads to profiles that are not in $\mathbb{P}$ and have different Plurality winners.}\end{remark}

To see that we cannot simply drop Homogeneity from the assumption of the theorem, consider the following example, whose verification is obvious (given that in a profile allowing ties, we cannot determine the parity of the number of non-indifferent voters from the margin information).

\begin{fact}\label{EvenOdd} Let $F$ be the following voting rule that violates Homogeneity: if the number of voters who are not fully indifferent is even, then use some margin-based voting rule, e.g., Beat Path (\citealt{Schulze2011}); if that number is odd, then use some other such rule, e.g., Ranked Pairs (\citealt{Tideman1987}). Then $F$ satisfies Tiebreaking Compensation and Neutral Indifference, but it is not margin-based.
\end{fact}

\section{Characterization modulo head-to-head}\label{HeadToHead}

We now return to the distinction introduced in Section~\ref{Intro} between margin-based voting rules and the larger class of head-to-head voting rules. We already know as a consequence of Theorem~\ref{StrictWeakChar}.\ref{StrictWeakChar1} that a homogeneous, head-to-head voting rule on the domain of all profiles is margin-based if and only if it satisfies Tiebreaking Compensation and Neutral Indifference. However, once we assume that the voting rule is head-to-head, we can provide an alternative axiomatization of margin-based rules using only Neutral Indifference and a strengthening of Neutral Reversal in the context of possibly nonlinear ballots called Nonlinear Neutral Reversal, introduced in the next subsection.

In this section, we work with voting rules on the domain of all profiles (for $\mathsf{LOBI}$, we refer the reader back to Theorem \ref{StrictWeakChar}.\ref{StrictWeakChar2}). For notational convenience, where $S_1, \dots, S_n$ partition $X$, \[S_1>S_2>\dots > S_n\] represents the strict weak order in which the $S_i$'s are the indifference classes, and all candidates in $S_1$ are strictly preferred to all candidates in $S_2$, etc.

\subsection{Nonlinear Neutral Reversal} 

The new axiom we will use in our characterization of margin-based rules relative to the class of head-to-head rules is a strengthening of Neutral Reversal. It applies not only when two voters submit reversed linear orders but also when two voters submit reversed rankings with ties.

\begin{definition} A voting rule $F$ satisfies Nonlinear Neutral Reversal if for any ${\mathbf{P},\mathbf{P}'\in\mathrm{dom}(F)}$, if $\mathbf{P}'$ is obtained from $\mathbf{P}$ by adding two voters with strict weak orders $P$ and $P^{-1}$, respectively, then $F(\mathbf{P})=F(\mathbf{P}')$. 
\end{definition}

\noindent To see that this axiom is indeed stronger than Neutral Reversal, consider the following.

\begin{fact} Let $F$ be the version of the Borda Count for profiles of strict weak orders defined as follows: from a strict weak order $P$, a candidate $x$ receives $|\{y\in X(\mathbf{P})\mid (x,y)\in P\}|$ points, and the candidates with the most points win. Then adding the pair of rankings $\{a\}> X(\mathbf{P})\setminus\{a\}$ and $X(\mathbf{P})\setminus\{a\}>\{a\}$ to a profile gives 
 $|X(\mathbf{P})|-1$ points to $a$ and $1$ point to everyone else, which can change the set of winners, so $F$ does not satisfy Nonlinear Neutral Reversal. Yet it is easy to see that $F$ does satisfy Neutral Reversal.\end{fact}

It is also easy to see that Neutral Indifference and Nonlinear Neutral Reversal are independent, as follows.

\begin{fact} $\,$
    \begin{enumerate}
    \item The version of Minimax based on winning votes satisfies Neutral Indifference but violates Nonlinear Neutral Reversal.
    \item Let $F$ be the voting rule from Fact \ref{EvenOdd}, which uses the parity of the number of voters. Then $F$ satisfies Nonlinear Neutral Reversal (since two voters are added for this axiom, maintaining parity) but violates Neutral Indifference (since one voter is added, changing parity).
    \end{enumerate}
\end{fact}

\subsection{Proof of characterization}

Let us now prove that Nonlinear Neutral Reversal and Neutral Indifference characterize margin-based rules relative to the class of head-to-head rules.

\begin{theorem}\label{ModuloH2H} For any head-to-head voting rule $F$ on the domain of all profiles, $F$ is margin-based if and only if $F$ satisfies  Nonlinear Neutral Reversal and Neutral Indifference.
\end{theorem}

\noindent The left-to-right direction of the `if and only if' is obvious. The right-to-left direction is an immediate corollary of Proposition~\ref{C2Char} characterizing C2 voting rules among head-to-head voting rules using Neutral Indifference and the following proposition characterizing margin-based voting rules among C2 voting rules using Nonlinear Neutral Reversal.

\begin{proposition}\label{MarginCharModuloC2}
    A C2 voting rule on the domain of all profiles is margin-based if and only if it satisfies Nonlinear Neutral Reversal.
\end{proposition}
\begin{proof}
For the non-trivial direction, it is enough to show the following:
  for any profiles $\mathbf{P}, \mathbf{P}'$, if $\margin_{\mathbf{P}} = \margin_{\mathbf{P}'}$, then there are profiles $\mathbf{Q}$ and $\mathbf{Q}'$ such that: 
  \begin{enumerate}
  \item $\mathbf{Q}$ (resp.~$\mathbf{Q}'$) is obtained by adding (possibly zero) (possibly nonlinear) reversal pairs of voters to $\mathbf{P}$ (resp.~$\mathbf{P}'$);
  \item $\#_\mathbf{Q}=\#_{\mathbf{Q}'}$.
  \end{enumerate}
For then by part 1 and Nonlinear Neutral Reversal, $F(\mathbf{P})=F(\mathbf{Q})$ and $F(\mathbf{P}')=F(\mathbf{Q}')$, and by part~2 and C2, $F(\mathbf{Q})=F(\mathbf{Q}')$, so $F(\mathbf{P})=F(\mathbf{P}')$, showing that $F$ is margin-based.

  Let $\mathbf{P}$ and $\mathbf{P}'$ satisfying the requirement be given with $X := X(\mathbf{P}) = X(\mathbf{P}')$. Recall that $X$ is finite. For any doubleton $D \subseteq X$, we define two profiles $\mathbf{A}_{D}$ and $\mathbf{A}'_{D}$ as follows. Let $a$ and $b$ be the smaller and the greater elements of $D$, respectively, according to the alphabetic ordering of $X$.
  Then let 
  \[
  d_{D} = \#_{\mathbf{P}'}(a, b) - \#_{\mathbf{P}}(a, b).
  \]
  Note that since by assumption $\mathcal{M}_{\mathbf{P}} = \mathcal{M}_{\mathbf{P}'}$, $d_D$ is also equal to $\#_{\mathbf{P}'}(b, a) - \#_{\mathbf{P}}(b, a)$. So, as long as $D = \{x, y\}$, $d_D = \#_{\mathbf{P}'}(x, y) - \#_{\mathbf{P}}(x, y)$.
  
  Now define $\mathbf{A}_{D}$ and $\mathbf{A}'_{D}$ as follows:
  \begin{itemize}
    \item If $d_{D} = 0$, let both $\mathbf{A}_{D}$ and $\mathbf{A}'_{D}$ be empty.\footnote{Our formal definition of profiles requires at least one voter, but we will not apply any voting rule to these $\mathbf{A}_{D}$ and $\mathbf{A}'_{D}$. Other concepts such as margin and winning votes apply easily to empty profiles.} 

    \item If $d_{D} > 0$, then intuitively we need to add to $\mathbf{P}$ this many voters who rank $a$ above $b$ and also this many voters who rank $b$ above $a$ (since we need to maintain the margin between $a$ and $b$), all while introducing as few side effects as possible. To do this, let $L$ be $\{a\} > \{b\} > X \setminus D$. Then let $\mathbf{A}_{D}$ consist of $2d_{D}$ fresh voters such that for each of $L$ and $L^{-1}$, there are exactly $d_{D}$ voters with that ballot.

    For $\mathbf{A}'_{D}$, first let $T$ be $D > X \setminus D$. Then let $\mathbf{A}'_{D}$ consist of $2d_{D}$ fresh voters such that for each of $T$ and $T^{-1}$, there are exactly $d_{D}$ voters with that ballot.  
  
    \item If $d_{D} < 0$, switch the role of $\mathbf{A}_{D}$ and $\mathbf{A}'_{D}$ in the above construction while using $-d_{D}$ pairs of fresh voters for each profile. That is, $\mathbf{A}_{D}$ will use $-d_D$ pairs of $T$ and $T^{-1}$ ballots while $\mathbf{A}'_{D}$ will use $-d_D$ pairs of $L$ and $L^{-1}$.

    \item It is instructive to do a little counting here: 
    \begin{itemize}
      \item For any $x, y \in X \setminus D$, plainly $\#_{\mathbf{A}_{D}}(x, y) = \#_{\mathbf{A}'_{D}}(x, y) = 0$ since they are always in the same equivalence class.

      \item For any $x \in D$ and $y \in X \setminus D$, we have $\#_{\mathbf{A}_{D}}(x, y) = \#_{\mathbf{A}'_{D}}(x, y) = |d_{D}|$ (by $L$ and $T$) and also $\#_{\mathbf{A}_{D}}(y, x) = \#_{\mathbf{A}'_{D}}(y, x) = |d_{D}|$ (by $L^{-1}$ and $T^{-1}$).

    \item When $d_{D} > 0$, 
    \[
    \#_{\mathbf{A}_{D}}(a, b) = \#_{\mathbf{A}_{D}}(b, a) = d_{D}\mbox{ while }\#_{\mathbf{A}'_{D}}(a, b) = \#_{\mathbf{A}'_{D}}(b, a) = 0.
    \] 
    Similarly, when $d_{D} < 0$, 
    \[
    \#_{\mathbf{A}_{D}}(a, b) = \#_{\mathbf{A}_{D}}(b, a) = 0\mbox{ while }\#_{\mathbf{A}'_{D}}(a, b) = \#_{\mathbf{A}'_{D}}(b, a) = -d_{D}.
    \] 
    Also, when $d_{D} = 0$, trivially 
    \[
    \#_{\mathbf{A}_{D}}(a, b) = \#_{\mathbf{A}_{D}}(b, a) = 0\mbox{ and }\#_{\mathbf{A}'_{D}}(a, b) = \#_{\mathbf{A}'_{D}}(b, a) = 0.
    \] 
    Taking all three cases into account, we always have that 
    \[
    \#_{\mathbf{A}_{D}}(a, b) - \#_{\mathbf{A}'_{D}}(a, b) = \#_{\mathbf{A}_{D}}(b, a) - \#_{\mathbf{A}'_{D}}(b, a) = d_{D}.
    \]
    \end{itemize}
    The counting above shows that for any $x, y \in X$ with $x\neq y$:
    \begin{align*}
      \#_{\mathbf{A}_{D}}(x, y) - \#_{\mathbf{A}'_{D}}(x, y) =
      \begin{cases}
        0 & \text{ if } \{x, y\} \not= D \\
        d_{D} & \text{ if } \{x, y\} = D.
      \end{cases}
    \end{align*}
  \item Also notice that both $\mathbf{A}_{D}$ and $\mathbf{A}'_{D}$ consist solely of possibly nonlinear reversal pairs.
  \end{itemize}

  Now let $\mathbf{Q}$ (resp.~$\mathbf{Q}'$) be the disjoint union of $\mathbf{P}$ (resp.~$\mathbf{P}'$) and $\mathbf{A}_{D}$ (resp.~$\mathbf{A}'_{D}$) for all doubletons $D \subseteq X$. Then $\#_\mathbf{Q} = \#_{\mathbf{Q}'}$. This is because, for any $x \not= y \in X$, 
  \begin{align*}
    &\#_\mathbf{Q}(x, y) - \#_{\mathbf{Q}'}(x, y) \\
    &= \left(\#_{\mathbf{P}}(x, y) +
      \sum_{D \subseteq X, |D| = 2}\#_{\mathbf{A}_{D}}(x, y)
      \right)
      - \left(\#_{\mathbf{P}'}(x, y) +
      \sum_{D \subseteq X, |D| = 2}\#_{\mathbf{A}'_{D}}(x, y)
      \right) \\
    &= \#_{\mathbf{P}}(x, y) - \#_{\mathbf{P}'}(x, y) +
      \sum_{D \subseteq X, |D| = 2}(\#_{\mathbf{A}_{D}}(x, y) -
      \#_{\mathbf{A}'_{D}}(x, y)) \\
    &= -d_{\{x, y\}} + (\#_{\mathbf{A}_{\{x, y\}}}(x, y) - \#_{\mathbf{A}'_{\{x, y\}}}(x, y)) \\
    &= -d_{\{x, y\}} + d_{\{x, y\}} \\
    &= 0.\qedhere
  \end{align*} 
 \end{proof}

 Note that we cannot drop the assumption of head-to-head from Theorem~\ref{ModuloH2H}, as shown by the following.

\begin{fact} The Positive/Negative voting rule satisfies Nonlinear Neutral Reversal and Neutral Indifference, but it is not margin-based.
\end{fact}

Combining Proposition~\ref{C2Char} and Proposition~\ref{MarginCharModuloC2}, we obtain the implications shown in Figure~\ref{ImpFig}.

\begin{figure}[h]
\begin{center}
\begin{tikzpicture}
      \node  at (0,0) (a) {\textbf{C2}}; 
      \node at (0,3) (b) {\textbf{head-to-head}}; 
      \node  at (0,-3) (d) {\textbf{margin-based}};

      \path[->,draw,thick] (b) to node[fill=white]  {Neutral Indifference} (a);

      \path[->,draw,thick,] (a) to node[fill=white] {Nonlinear Neutral Reversal}  (d);

    \end{tikzpicture}
    \end{center}
    \caption{Implications from head-to-head to C2 to margin-based modulo additional axioms.}\label{ImpFig}
\end{figure}
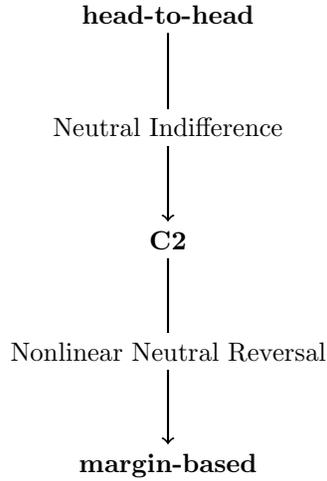

\section{Conclusion}\label{Conclusion}

For voting rules that accept profiles of ranked ballots, the invariance property of being \textit{margin-based} can obviously be regarded as a property of informational parsimony: only the information about head-to-head margins of victory or loss is needed to determine the outcome. What is less obvious is that being margin-based is equivalent to a conjunction of axioms with normative significance, drawing on certain notions of equality or balancing between voters with different preferences. Our axiom of Preferential Equality expresses one such notion of equality: either of two voters changing $xy$ to $yx$ in their ranking should have the same effect as if the other voter did so. And Neutral Reversal expresses one such notion of balancing: two voters with fully reversed rankings balance each other out, so adding them does not change the election outcome. Our axiom of Tiebreaking Compensation, which implies Preferential Equality and almost implies Neutral Reversal (it does together with Neutral Indifference), also incorporates these notions of equality and balancing. 

It is an interesting open question whether other invariance properties can be proved equivalent to combinations of natural, normatively significant axioms. For example, can the class of head-to-head voting rules be characterized in terms of Preferential Equality together with some independent axioms?\footnote{To apply the strategy used in this paper, we would need to show that if two profiles share the same head-to-head information, then there is a way to turn them into the same profile using operations preserving the head-to-head information, and in particular, the number of voters. This seems to require a detailed analysis of when and how the same margin matrix can be realized by multiple profiles of the same size.} While the class of head-to-head rules is larger than that of margin-based rules, we can also go in the other direction and consider more restricted classes. For example, many of the well-known margin-based rules (recall Footnote~\ref{CondorcetRules}) satisfy the stronger property of \textit{ordinal margin invariance} (see \citealt{Holliday2024}), according to which we do not need to know the exact sizes of margins but only whether the margin of $x$ over $y$ is greater than that of $x'$ over $y'$. Or one could go all the way to Fishburn's \citeyearpar{Fishburn1977} class of C1 rules, according to which all we need to know about the margins is their signs (positive, zero, or negative). It is doubtful that increasingly strong invariance properties will continue to have normatively compelling equivalent axiomatizations, since at some point (arguably by the time we reach C1) the invariance properties appear clearly too strong. Indeed, the project of seeking normatively perspicuous re-axiomatizations of invariance properties may guide us toward the boundaries between desirable and undesirable invariance properties for voting rules.

\subsection*{Acknowledgements}

We thank the audiences at the COMSOC Video Seminar in January 2025 and at the Special Session on the Mathematics of Decisions, Elections, and Games at the Joint Mathematics Meeting 2025, as well as the anonymous referees for \textit{Social Choice and Welfare} for their helpful feedback.

\bibliographystyle{plainnat}
\bibliography{majority}

\appendix

\section{Appendix: Allowing Noncomparability}\label{App}

In this appendix, we consider axiomatizations of margin-based rules on other domains of voter preferences. As noted in Section~\ref{StrictWeakProfs}, in many actual election scenarios, voters are not required to rank all the candidates, but voters are required to \textit{linearly} order any candidates they do rank. In voting theory, such ballots are often interpreted as a special case of weak orders by putting all unranked candidates in an indifference class at the bottom. An alternative interpretation of truncated linear ballots, which may be appropriate in some circumstances, treats each unranked candidate as noncomparable to all other candidates (this option is available on some voting websites). Such a noncomparability interpretation of unranked candidates may also make sense for some elections in which ballots allow both ties and unranked candidates (as also allowed on some voting websites).

\begin{definition}
  Let $X$ be a nonempty set and $R$ be a binary relation on $X$. This $R$ is to be interpreted as a weak preference relation on $X$. We write $I(R)$, $P(R)$, and $N(R)$ for the indifference, strict preference, and noncomparability relation induced by $R$,  defined as follows: for any $x, y \in X$,
  \begin{itemize}
    \item $x I(R) y$ iff $x R y$ and $y R x$;
    \item $x P(R) y$ iff $x R y$ but not $y R x$;
    \item $x N(R) y$ iff neither $x R y$ nor $y R x$.
  \end{itemize}
  We also change our definition of a profile to mean a function $\mathbf{R}: V \to \wp(X^2)$ where $V$ is a nonempty finite subset of $\mathcal{V}$, also denoted by $V(\mathbf{R})$, and $X$ is a nonempty finite subset of $\mathcal{X}$ with $|X|\geq 2$, also denoted by $X(\mathbf{R})$. For any profile $\mathbf{R}$, its \emph{head-to-head counting functions} $\#P_{\mathbf{R}}, \#I_{\mathbf{R}}, \#N_{\mathbf{R}}: X(\mathbf{R})^2 \to \mathbb{N}$ are defined by 
  \begin{itemize}
    \item $\#P_{\mathbf{R}}(x, y) = |\{i \in V(\mathbf{R}) \mid xP(\mathbf{R}(i))y\}|$,
    \item $\#I_{\mathbf{R}}(x, y) = |\{i \in V(\mathbf{R}) \mid xI(\mathbf{R}(i))y\}|$, and
    \item $\#N_{\mathbf{R}}(x, y) = |\{i \in V(\mathbf{R}) \mid xN(\mathbf{R}(i))y\}|$.
  \end{itemize}
  Then the \emph{total head-to-head information} $\mathcal{H}(\mathbf{R})$ of $\mathbf{R}$ is simply the tuple $(\#P_{\mathbf{R}}, \#I_{\mathbf{R}}, \#N_{\mathbf{R}})$ collecting the three counting functions. The margin function is defined as before: \[\margin_{\mathbf{R}}(x, y) = \#P_{\mathbf{R}}(x, y) - \#P_{\mathbf{R}}(y, x).\]

  A \emph{voting rule} $F$ is a function whose domain $\mathrm{dom}(F)$ is some set of profiles. It is \emph{head-to-head} (resp.~\emph{margin-based}) if for any $\mathbf{R}, \mathbf{R}' \in \mathrm{dom}(F)$, if $\mathcal{H}(\mathbf{R}) = \mathcal{H}(\mathbf{R}')$ (resp.~$\margin_{\mathbf{R}} = \margin_{\mathbf{R}'}$), then $F(\mathbf{R}) = F(\mathbf{R}')$.
  \end{definition}

    The alternative interpretations of truncated ballots can be formally represented as follows.
  \begin{definition}
  For any nonempty set $X$ and any binary relation $R$ on $X$, the \emph{side-noncomparability class} $S(R)$ of $R$ is $\{x \in X \mid \forall y \in X \setminus \{x\}, x N(R) y\}$. Intuitively, this is the set of candidates that do not appear on the truncated ballot, or the set of all candidates if there is only one candidate on the truncated ballot. 
  \begin{itemize}
    \item A binary relation $R$ on $X$ is a \textit{linear order plus side-noncomparability} (`LOSN' for short) if $R$ is reflexive, transitive, and for all $x,y \in X \setminus S(R)$ with $x\neq y$, either $xP(R)y$ or $yP(R)x$.
    \item A binary relation $R$ on $X$ is a \textit{weak order plus side-noncomparability} (`WOSN' for short) if $R$ is reflexive, transitive, and for all $x,y\in X \setminus S(R)$, either $xRy$ or $yRx$.
  \end{itemize}
  Let $\mathsf{LOSN}$ (resp.~$\mathsf{WOSN}$) be the set of all profiles $\mathbf{R}$ such that for any $i \in V(\mathbf{R})$, $\mathbf{R}(i)$ is an LOSN (resp.~WOSN) on $X(\mathbf{R})$. We call $S(\mathbf{R}(i))$ voter $i$'s \textit{side-noncomparability class}.
\end{definition}

To obtain an analogue of Theorem \ref{StrictWeakChar} for these domains, it suffices to add just one axiom (cf.~the ``Pure'' version of Tiebreaking Compensation from Remark~\ref{Replacements}).

\begin{definition} A voting rule $F$ satisfies Comparable Compensation if for any $\mathbf{R},\mathbf{R}'\in \mathrm{dom}(F)$, if in $\mathbf{R}$ there are two voters $i$ and $j$ with the \textit{same} ranking with a nonempty side-noncomparability class $S$, and $\mathbf{R}'$ is obtained from $\mathbf{R}$ by $i$ ranking all candidates in $S$ \textit{above} those in $X(\mathbf{R})\setminus S$ while ranking those within $S$ according to a linear order $L$, and $j$ ranking all candidates in $S$  \textit{below} all those in $X(\mathbf{R})\setminus S$ while ranking those within $S$ according to $L^{-1}$, then $F(\mathbf{R})=F(\mathbf{R}')$.
\end{definition}

Clearly any margin-based voting rule satisfies Comparable Compensation. 

\begin{theorem} Let $F$ be a voting rule satisfying Homogeneity.
\begin{enumerate}
\item\label{LOSNChar} If the domain of $F$ is $\mathsf{LOSN}$, then $F$ is margin-based if and only if $F$ satisfies Preferential Equality, Comparable Compensation, and Neutral Reversal.
\item\label{WOSNChar} If the domain of $F$ is $\mathsf{WOSN}$, then $F$ is margin-based if and only if $F$ satisfies Tiebreaking Compensation, Comparable Compensation, and Neutral Indifference.
\end{enumerate}
\end{theorem}

\begin{proof} If a profile contains a ranking with nonempty side-noncomparability class $S$, duplicate the profile, so we have two copies of the ranking, and then linearize $S$ in opposite ways above and below the other candidates. By Homogeneity and  Comparable Compensation, this does not change the output of $F$. In this way, we can change any two initial profiles into profiles with no side-noncomparability. Then for part \ref{LOSNChar} we reason exactly as in the proof of Theorem \ref{LinChar} using Preferential Equality, and for part \ref{WOSNChar} we reason exactly as in the proof of Theorem \ref{StrictWeakChar}.\ref{StrictWeakChar1} using Tiebreaking Compensation and Neutral Indifference.\end{proof}

To obtain analogues of Theorem \ref{ModuloH2H}, note first that the Nonlinear Neutral Reversal axiom transfers to the current context seamlessly, as both $\mathsf{LOSN}$ and $\mathsf{WOSN}$ are closed under reversal: if $R$ is a LOSN (resp.~a WOSN), then $R^{-1}$ is also a LOSN (resp.~a WOSN). The normative appeal of the axiom should also be the same in the current context. 

For $\mathsf{LOSN}$, we also use the following axiom that is essentially Neutral Indifference for $\mathsf{LOSN}$.

\begin{definition}
  A voting rule $F$ satisfies \emph{Neutral Blankness} if for all $\mathbf{R}, \mathbf{R}' \in \mathrm{dom}(F)$ such that $\mathbf{R}'$ is obtained from $\mathbf{R}$ by adding a voter whose side-noncomparability class is $X(\mathbf{R})$, then $F(\mathbf{R}) = F(\mathbf{R}')$.
\end{definition}
\noindent As we mentioned above, when $S(R) = X$, what $R$ represents is either the empty ballot where no candidate is ranked or the ballot that ranks only one candidate, which, under our current interpretation, is no more informative than the empty ballot, since we must not take ranked candidates to be better (or worse) than unranked candidates.

For $\mathsf{WOSN}$, a voter might rank some candidates as tied and leave all others unranked. In this case, all the indifference classes---namely, the set of ranked candidates plus each singleton set of an unranked candidate---are noncomparable to each other under our current interpretation of truncated ballots; the only difference between them is that the indifference class of ranked candidates may be larger than that of an unranked candidate. It seems wrong to favor an indifference class merely because it is larger (or smaller) than some other indifference classes.
Thus, we introduce the following axiom, noting that these ballots where all ranked candidates are tied are mathematically represented by self-reversing, and thus self-canceling, binary relations.
\begin{definition}
  A binary relation $R$ is \emph{self-reversing} if $R = R^{-1}$. A voting rule $F$ satisfies \emph{Neutral Self-Reversal} if for any $\mathbf{R}, \mathbf{R}' \in \mathrm{dom}(F)$, if $\mathbf{R}'$ is obtained from $\mathbf{R}$ by adding a voter whose ballot is self-reversing, then $F(\mathbf{R}) = F(\mathbf{R}')$.
\end{definition}
Since the blank ballot is self-reversing, Neutral Self-Reversal strengthens Neutral Blankness.

\begin{theorem} \label{thm:extended-characterization} Let $F$ be a head-to-head voting rule on a domain $D$ of profiles.
  \begin{enumerate}
    \item If $D = \mathsf{LOSN}$, then $F$ is margin-based iff $F$ satisfies Neutral Blankness and Nonlinear Neutral Reversal.
    \item If $D = \mathsf{WOSN}$, then $F$ is margin-based iff $F$ satisfies Neutral Self-Reversal and Nonlinear Neutral Reversal.
  \end{enumerate}
\end{theorem}
\begin{proof}
    Consider the case where $D = \mathsf{LOSN}$. Take any two profiles $\mathbf{R}, \mathbf{Q} \in \mathsf{LOSN}$ such that $\margin_{\mathbf{R}} = \margin_{\mathbf{Q}}$. Then for each $x,y \in X := X(\mathbf{R}) = X(\mathbf{Q})$ with $x\neq y$, we denote by $R_{x>y}$ the LOSN where the only strict preference is that $x$ is strictly preferred over $y$. Then $R^{-1}_{x>y} = R_{y>x}$. To iterate over the doubletons from $X$, let $\prec$ be the alphabetic order on $X$. Now we construct profiles $\mathbf{R}_1$ and $\mathbf{Q}_1$ as follows. Initially, let $\mathbf{R}_1 = \mathbf{R}$ and $\mathbf{Q}_1 = \mathbf{Q}$, and then for each $x ,y \in X$ with $x\prec y$:
    \begin{itemize}
      \item if $n = \#P_{\mathbf{R}}(x, y) - \#P_{\mathbf{Q}}(x, y) > 0$, then since $\margin_{\mathbf{R}} = \margin_{\mathbf{Q}}$,  $\#P_{\mathbf{R}}(y, x) - \#P_{\mathbf{Q}}(y, x)$ is also $n$, and in this case add to $\mathbf{Q}_1$ $n$ pairs of fresh voters where each pair votes $R_{x > y}$ and $R_{y > x}$;
      \item if $n = \#P_{\mathbf{Q}}(x, y) - \#P_{\mathbf{R}}(x, y) > 0$, then since $\margin_{\mathbf{R}} = \margin_{\mathbf{Q}}$,  $\#P_{\mathbf{Q}}(y, x) - \#P_{\mathbf{R}}(y, x)$ is also $n$, and in this case add to $\mathbf{R}_1$ $n$ pairs of fresh voters where each pair votes $R_{x > y}$ and $R_{y > x}$.
    \end{itemize}
    It is clear that $\mathbf{R}_1$ and $\mathbf{Q}_1$ have not only the same margins but also the same $\#P$ function. They may have different numbers of voters, but we can also add to the smaller profile enough fresh voters with ballots having no strict preferences (i.e., the identity relation) and obtain $\mathbf{R}_2$ and $\mathbf{Q}_2$. Then $\mathcal{H}(\mathbf{R}_2) = \mathcal{H}(\mathbf{Q}_2)$. Using the assumption that $F$ is head-to-head and satisfies Nonlinear Neutral Reversal and Neutral Blankness, we have $F(\mathbf{R}) = F(\mathbf{Q})$.

    For the case where $D = \mathsf{WOSN}$, again we take profiles $\mathbf{R}$ and $\mathbf{Q}$ from $\mathsf{WOSN}$ such that $\margin_{\mathbf{R}} = \margin_{\mathbf{Q}}$, and we want to show that $F(\mathbf{R}) = F(\mathbf{Q})$. As in the above case, we construct $\mathbf{R}_1$ and $\mathbf{Q}_1$ in the same way. Now we also need to modify their $\#I$ functions to be the same. This can be done by the self-reversing WOSN ballot $R_{x = y}$ where there is no strict preference and the only non-trivial indifference is between $x$ and $y$. In other words, $R_{x = y}$ is the union of the identity relation on $X(\mathbf{R})$ and $\{(x, y), (y, x)\}$. We obtain $\mathbf{R}_2$ and $\mathbf{Q}_2$ by adding, for each $x,y\in X$ with $x\prec y$, $n = \#I_{\mathbf{R}_1}(x, y) - \#I_{\mathbf{Q}_1}(x, y)$ many fresh $R_{x = y}$ ballots to $\mathbf{Q}_1$ if $n$ is positive, or $m =\#I_{\mathbf{Q}_1}(x, y) - \#I_{\mathbf{R}_1}(x, y)$ many fresh $R_{x = y}$ ballots to $\mathbf{R}_1$ if $m$ is positive. Then both the $\#P$ function and the $\#I$ function of $\mathbf{R}_2$ and $\mathbf{Q}_2$ are the same. To make the $\#N$ functions the same as well, we only need to add blank ballots to the smaller profile and obtain $\mathbf{R}_3$ and $\mathbf{Q}_3$ of equal size. Since $F$ is head-to-head and $\mathcal{H}(\mathbf{R}_3) = \mathcal{H}(\mathbf{Q}_3)$, we have $F(\mathbf{R}_3) = F(\mathbf{Q}_3)$. Then since $F$ satisfies Neutral Self-Reversal and Nonlinear Neutral Reversal, $F(\mathbf{R}) = F(\mathbf{Q})$.
\end{proof}
\begin{corollary}
 Let $F$ be a head-to-head and homogeneous voting rule whose domain is either $\mathsf{LOSN}$ or $\mathsf{WOSN}$. Then $F$ is margin-based iff $F$ satisfies Nonlinear Neutral Reversal. 
\end{corollary}
\begin{proof}
    Clearly, a homogeneous $F$ that satisfies Nonlinear Neutral Reversal also satisfies Neutral Self-Reversal, since to add a self-reversal ballot $R$ without changing the output, we can first double the profile, add $R$ twice, which now counts as a reversal pair, and then ``divide the profile by~2''.
\end{proof}

\end{document}